\documentclass{article}

\PassOptionsToPackage{numbers, compress}{natbib}
\usepackage[preprint]{neurips_2022}

\usepackage[utf8]{inputenc} % allow utf-8 input
\usepackage[T1]{fontenc}    % use 8-bit T1 fonts
\usepackage{hyperref}       % hyperlinks
\usepackage{url}            % simple URL typesetting
\usepackage{booktabs}       % professional-quality tables
\usepackage{amsfonts}       % blackboard math symbols
\usepackage{nicefrac}       % compact symbols for 1/2, etc.
\usepackage{microtype}      % microtypography
\usepackage{xcolor}         % colors
\usepackage{amsmath}
\usepackage{amsthm}
\usepackage{amsfonts}
\usepackage{tcolorbox}
\usepackage{physics}
\usepackage[ruled,linesnumbered,vlined]{algorithm2e}
\usepackage{parskip}
\usepackage{graphicx}
\usepackage[shortlabels]{enumitem}
\usepackage{thmtools, thm-restate}

\SetCommentSty{mycommfont}

\SetKwInput{KwInput}{Input}                % Set the Input
\SetKwInput{KwOutput}{Output}

\tcbuselibrary{theorems}
\newtcbtheorem[number within=section]{definition}{Definition}%
{colback=gray!5,colframe=gray!100!black,fonttitle=\bfseries}{def}
\newcounter{saveequation}
\newenvironment{alignLetter}{
  \setcounter{saveequation}{\value{equation}}% Save the current equation number
  \setcounter{equation}{0}% Reset equation counter for this environment
  \align
}{
  \endalign
  \setcounter{equation}{\value{saveequation}}% Restore the saved equation number
}

\title{Trainability issues in quantum policy gradients}

\usepackage{authblk}
\author[1,2,3]{André Sequeira}
\author[1,2,3]{Luis Paulo Santos}
\author[1,2,3]{Luis Soares Barbosa}
\affil[1]{Department of Informatics, University of Minho, Braga, Portugal} 
\affil[2]{High Assurance Software Laboratory, INESC TEC, Braga, Portugal}
\affil[3]{Quantum Linear-optical computation group, International Nanotechnology Laboratory, Braga, Portugal}

\begin{document}

\maketitle

\begin{abstract}
  This research explores the trainability of Parameterized Quantum Circuit-based policies in Reinforcement Learning, an area that has recently seen a surge in empirical exploration. While some studies suggest improved sample complexity using quantum gradient estimation, the efficient trainability of these policies remains an open question. Our findings reveal significant challenges, including standard Barren Plateaus with exponentially small gradients and gradient explosion. These phenomena depend on the type of basis-state partitioning and the mapping of these partitions onto actions. For a polynomial number of actions, a trainable window can be ensured with a polynomial number of measurements if a contiguous-like partitioning of basis-states is employed. These results are empirically validated in a multi-armed bandit environment.
\end{abstract}

%\tableofcontents

\section{Introduction}\label{sec: introduction}
Variational Quantum Algorithms (VQAs), emerging as a cornerstone in the Noisy Intermediate Scale Quantum (NISQ) era, present a novel approach to overcoming the limitations inherent in quantum computing, such as restricted qubit availability and noise related constraints on circuit depth. Initially proposed as universal computation models \citep{biamonte_universal_2021}, VQAs operate through a synergy of quantum and classical mechanisms. They utilize a Parameterized Quantum Circuit (PQC) where the parameters are fine-tuned via a classical optimization routine to achieve the global optimum of a specified objective function \citep{cerezo_variational_2021}. Despite the theoretical allure of VQAs, their practical efficiency is often hampered by the so-called barren plateau (BP) phenomenon, a critical challenge in quantum optimization \citep{mcclean_barren_2018}. This phenomenon, characterized by the exponential suppression of the gradients' magnitude with an increasing number of qubits, requires an exponentially large number of measurements to allow the algorithm to effectively navigate through the optimization landscape.

The BP phenomenon pose a significant hurdle, not only in gradient-based but also in gradient-free optimization approaches, where cost concentration emerges as a parallel challenge \cite{arrasmith_effect_2021}. Understanding and mitigating the occurrence of BPs in specific VQAs is thus vital for harnessing any potential quantum advantage. Several factors contribute to the emergence of BPs, including deep and random quantum circuits \cite{mcclean_barren_2018}, PQCs adhering to a volume law in entanglement entropy \citep{leone_practical_2022} etc. The work of Cerezo et al. \citep{cerezo_cost_2021} particularly highlights the dependence of the BP phenomenon on the locality of the cost function, showing that local losses measured on a logarithmic number of qubits can retain trainability in shallow circuits \citep{rudolphTrainabilityBarriersOpportunities2023}.

Further complicating the picture, conventional machine learning cost functions like the mean squared error, negative log likelihood, and KL-divergence have been shown to lead to BPs \citep{thanasilp_subtleties_2021}. BPs are typically characterized by the scaling of the variance of partial derivatives of the cost function, which diminishes exponentially with the number of qubits \citep{mcclean_barren_2018}. This scaling results in gradients increasingly concentrating around zero, making optimization exceedingly difficult. Another approach to characterize a BP is through the study of cost concentration \citep{arrasmith_effect_2021}, where cost differences between randomly selected points in the landscape show an exponential concentration with increasing qubits. In addition, the Fisher Information Matrix (FIM) spectrum, as explored in the work of Abbas et al. \citep{abbas_power_2021}, offers valuable insights into the flatness of the loss landscape in the presence of BPs, with the eigenvalues of the FIM becoming exponentially small as the number of qubits increases.

Recent studies have expanded the application of VQAs to Reinforcement Learning (RL), showing promising results \citep{skolik_quantum_2022, jerbi_variational_2021, sequeira_policy_2023}. Notably, the work of Jerbi et al. \citep{jerbi_quantum_2022} demonstrated a quadratic optimization improvement in policy-based RL agents using PQC-based policies over classical agents. However, the trainability of these quantum policies, particularly in the face of BPs, remains an open question. In this context, our study aims to provide a deeper understanding of the trainability issues associated with PQC-based policies in RL, focusing on cost-function dependent BPs and their implications. We explore the challenges faced by specific variations of previously proposed \textit{Raw Policies}\cite{jerbi_variational_2021,meyer_quantum_2023}, and investigate their performance under various conditions. Our findings contribute to the ongoing research on optimizing PQC-based agents in quantum RL, addressing critical questions on the interplay between policy types, number of qubits, action-space size, and the presence of BPs and other trainability issues such as exploding gradients. This research not only advances our understanding of quantum RL but also sets the stage for future investigations into other types of PQC-based policies \cite{sequeira_policy_2023,jerbi_variational_2021}, thereby unlocking the full potential of quantum computing in machine learning applications.

\subsection*{Related work}\label{subsec: related work}
Recent advancements have been documented concerning the application of VQAs to RL. There has been considerable empirical evidence supporting the efficacy of VQAs in diverse benchmark environments, encompassing both value-based \citep{chen_variational_2020, skolik_quantum_2022} and policy-based \citep{jerbi_variational_2021, sequeira_policy_2023} RL paradigms. A significant contribution in this field was made by Jerbi et al. \citep{jerbi_quantum_2022}, who demonstrated a quadratic improvement in gradient estimation for optimizing policy-based RL agents using PQC-based policies compared to purely classical agents.
In another notable work, Cherrat et al. \citep{cherrat_quantum_2023} introduced quantum neural network architectures featuring orthogonal and compound layers for policy and value functions, notably devoid of BPs in the context of financial hedging. At the same time, Meyer et al. \citep{meyer_quantum_2023} posited that a global parity-based policy could provide more information to the agent and a more conducive optimization landscape. This proposition challenges the previously held belief that global measurements lead to flatter landscapes, implying further issues on trainability. The emerging divergence in these findings entails the need for further research to fully understand the impact of the BP phenomenon within RL, especially in the context of generalized PQC-based policies, as it may significantly influence optimization efficiency provided by gradient estimation.

This investigation centers on analyzing \textit{cost-function dependent barren plateaus} within the framework of policy-based RL, utilizing both local and global projector-based observables in conjunction with PQC-based policies. The primary objective of this study is to delineate variance limits for the gradient of the REINFORCE policy-dependent objective function \cite{williams_simple_1992}, especially under the assumption of a PQC-based policy. We re-examine two previously introduced policies, redefined here for enhanced clarity: 1) The \textit{Contiguous-like Born policy}, as referenced in \cite{jerbi_variational_2021}, derived from categorizing basis states into a contiguous set proportional to the action-space size, and 2) The \textit{Parity-like Born policy}, detailed in \cite{meyer_quantum_2023}, formulated through a recursive parity function applied to measured basis states.

\subsection*{Contributions}\label{subsec: contributions}
Our findings highlight that both contiguous and parity-like Born policies can potentially face extreme challenges in terms of trainability. On one side, the policy might encounter standard BPs characterized by exponentially vanishing gradients, while on the other, it may face issues of gradient explosion. These phenomena are heavily influenced by the locality of the observables employed that depend on the action-space size. For $n$ qubit policies estimated through $\mathcal{O}(\text{poly}(n))$ measurements, the \textit{contiguous-like Born policy} exhibits a trainable region at logarithmic depth $\mathcal{O}(\text{log}(n))$, assuming the action-space is of $\mathcal{O}(n)$ size. For a $\mathcal{O}(\text{poly}(n)$ number of actions, the policy enters a transition region where the locality of the observables increase but it is still possible to train under polynomially large number of measurements. Conversely, under the same conditions, the \textit{Parity-like Born policy} is untrainable, suffering from a BP. 
Beyond polynomially-sized action spaces, no policy can be trained using a polynomial number of measurements since the probability of measuring basis states becomes exponentially suppressed with the number of qubits. In such a scenario, the gradient behavior shifts towards exploding gradients due to the exponentially small probabilities.

The trainability of PQC-based policies was further analyzed by inspecting the FIM spectrum. It was observed that, under polynomially sized actions spaces, the FIM spectrum indeed reveals a BP for the \textit{Parity-like Born policy}, as FIM entries shrink exponentially with increasing qubits, resulting in a spectrum highly concentrated at zero, therefore characterizing a flat landscape. Outside polynomial action spaces, the FIM spectrum becomes less informative about BPs due to the exponentially small probabilities that induce large FIM entries, causing a shift in the spectrum with more eigenvalues concentrated away from zero.

Empirical validation of these results was achieved by examining the scaling of the variance of the log likelihood gradient, using the simplified two-design ansatz \cite{cerezo_cost_2021} for the PQC-based policy. Furthermore, the effect of the observables' global nature on a PQC-based agent's trainability was explored in the context of learning to select the optimal arm in a simulated multi-armed bandit environment. The observations confirmed that for a PQC-based agent with a polynomial number of actions, the contiguous-like Born policy is capable of learning the optimal arm, unlike the parity-like policy. However, when extending beyond a polynomial number of actions, both policies were unable to learn the optimal arm, in line with our theoretical predictions.

The rest of the document is organized as follows:
Section \ref{sec: QPG} introduces the policy gradient framework in RL and the intricacies behind PQC-based policies such as gradient estimation. Section \ref{sec: pg variance} establishes novel results forming the core of this work. It provides clear lower bounds for the policy gradient's variance. Section \ref{sec: numerical experiments} resorts to numerical experiments as an empirical validation of the theoretical predictions in Section \ref{sec: pg variance}. Finally, Section \ref{sec: conclusions} concludes the work and outlines future research directions.

\section{Quantum Policy Gradients}\label{sec: QPG}

Policy Gradient algorithms are designed to optimize a parameterized policy $\pi(a|s,\theta) = \mathbb{P}\{a_t = a|s_t = s, \theta_t = \theta \}$, where $\theta \in \mathbb{R}^k$ denotes the parameter vector with dimension $k$, $s$, $a$, and $t$ represent the state, action and the time step, respectively. The essence of this approach is to enable optimal action selection without relying on a value function, with the primary aim of maximizing a performance measure $J(\theta)$. This is achieved by applying gradient ascent to $J(\theta)$ as follows:

\begin{equation}
\theta_{i+1} = \theta_i + \eta \nabla_{\theta_i} J(\theta_i)
\end{equation}

where $\eta$ is the learning rate. For discrete and small action spaces, a Softmax-Policy is commonly used to balance exploration and exploitation. The Monte-Carlo policy gradient, known as REINFORCE, estimates the gradient from samples across $N$ trajectories of length $T$, or the horizon, under the parameterized policy. A known limitation of REINFORCE is the high variance of its gradient estimation due to the stochastic nature of sampling trajectories. This variance can negatively affect performance in complex settings. Introducing a baseline denoted by $b(s_t)$, such as the average return, can reduce the variance without having to increase the number of samples $N$. The baseline is subtracted from the returned value to stabilize the optimization process, as shown in Equation \eqref{eq: policy gradient baseline}.

\begin{equation}
\nabla_{\theta} J(\theta) = \frac{1}{N} \sum_{i=0}^{N-1}\sum_{t=0}^{T-1} (G_t(\tau_i) - b(s_{t_i})) \nabla_{\theta} \log \pi(a_{t_i} \lvert s_{t_i} , \theta)
\label{eq: policy gradient baseline}
\end{equation}

where $G_t(\tau_i)$ is the cumulative discounted return at time step $t$ in trajectory $\tau_i$. Throughout the rest of the paper, the baseline $b(s_t)$ is considered as the average return across all trajectories.

    \begin{equation} 
    b(s_t) = {1 \over {N}} \sum_{i=0}^{N-1} G_t(\tau_i)
    \label{eq: baseline}
    \end{equation}
 
In this work, we consider PQC-generated policies i.e., policies generated from Parameterized Quantum Circuits (PQC). Specifically we consider two variants of the \textit{raw policies} proposed in the literature and redefined here for enhanced clarity: 1)The \textit{Contiguous-like Born policy} \cite{jerbi_variational_2021} and 2) \textit{Parity-like Born policy} \cite{meyer_quantum_2023}. For completion, the Softmax-based PQC policy \cite{sequeira_policy_2023,jerbi_variational_2021} is also defined but addressing its trainability is outside of the scope of this work. Let us start with the most general definition of a Born policy.\\
\vspace{0.5cm}
\subsection{Born policy}
\begin{restatable}{defi}{bp}
    \label{def: born policy}
    Let $s \in \mathcal{S}$ be a state embedded in an $n$-qubit parameterized quantum state, $\ket{\psi(s,\theta)}$, where $\theta \in \mathbb{R}^k$. The probability associated to a given action $a \in A$ is given by: 
    \begin{equation}
        \pi(a|s,\theta) = \langle P_a \rangle_{s,\theta} = \bra{\psi(s,\theta)} P_a \ket{\psi(s,\theta)}
    \end{equation}
    where $P_a = \sum_{v \in V_a} \ket{v}\bra{v}$ is the projector into partition $V_a \subseteq V$ where $V = \{v_0, v_1, \dots, v_{2^n -1} \}$ is the set of eigenstates of an observable
    \begin{equation}
        O = \sum_{i=0}^{2^n - 1} \lambda_i \ket{v_i}\bra{v_i}
    \end{equation}  
    Moreover, $\bigcup_{a \in A} V_a = V$ and $V_a \cap V_{a'} = \emptyset$, for all $a \neq a'$.
\end{restatable}
Definition \ref{def: born policy} introduces the most general definition of a Born policy. However, there could be partitions that do not take into account every eigenstate of a given observable, as described above. In these scenarios, the probability associated to a given action would not be normalized as before since $\sum_{a \in A} P_a \neq I$. Moreover, since the goal of this work is the study of cost-function dependent BPs in quantum policy gradients, different partitions $V_a$, and the associated globality of the measurement should be further clarified. 
\begin{figure}[!htb]
    \centering
    \begin{minipage}[t]{\textwidth}
    \centering
    \includegraphics[width=\textwidth]{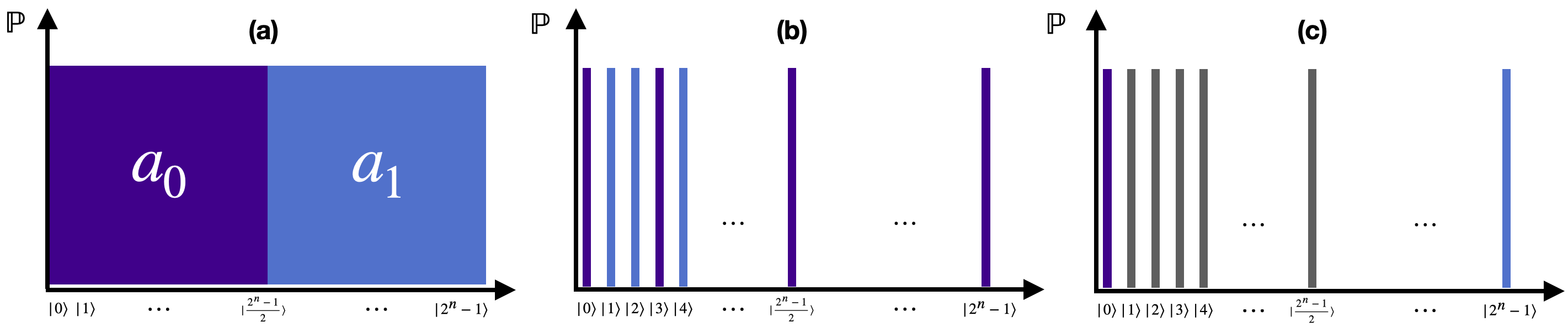}
    \end{minipage}%
    \hfill
    \caption[]{Partitions considered for the base case of $|A|=2$. (a) and (b) illustrates a Contiguous-like and parity-like partitioning, respectively, of all $2^n$ basis states. (c) Action-projector-like partitioning considering just two basis states.}
    \label{fig: partitions base case}
\end{figure}

Figure \ref{fig: partitions base case} illustrates different partitions considered throughout this work for $|A|=2$, which constitutes the base case in RL. 

\subsubsection*{Contiguous-like Born policy}
Consider the action space $A = \{ a_0, a_1\}$. The simplest partitioning that fits definition \ref{def: born policy}, would be to separate all basis states in half i.e., $V_{a_0} = \{|0\rangle , |1\rangle , |2\rangle \dots |\frac{2^n - 1}{2}\rangle\}$ and  $V_{a_1} = \{|\frac{2^n}{2}\rangle \dots |2^n - 1\rangle \}$. Such partitioning is illustrated in Figure \ref{fig: partitions base case}(a). In this case, even though every n-bit bitstring is considered to build the policy, a careful analysis of the partitioning indicates that it does not correspond to a global measurement. It is possible to assign a bitstring to its respective set by just measuring the first bit. If the bit is in state $|0\rangle$ (respectively, $|1\rangle$) it corresponds to the set $V_{a_0}$ (respectively, $V_{a_1}$). Thus, such assignment corresponds to a 1-local measurement, indeed.\\ 
In general, for an arbitrary number of actions $|A|\leq 2^n$ if we assign each bitstring to $|A|$ contiguous sets, then the measurement will actually be $(\log |A|)$-local, since to assign each bitstring $\log |A|$ bits are required to distinguish between the sets. As an example. let the the number of qubits be $n = 3$. The total number of bitstrings is $2^3 = 8$ , corresponding to the set $\{000, 001, 010, 011, 100, 101, 110, 111\}$. Suppose $|A| = 4$ with partition set $V = V_0 \cup V_1 \cup V_2 \cup V_3$. The number of bits needed to distinguish between sets is $\log_2(4) = 2$. Thus, the first 2 bits of each bitstring are considered to assign it to one of the 4 sets. Let $a$ be represented in its binary expansion. Then, the partitioning will be given by $V = \{000,001\} \cup \{010,011\} \cup \{100,101\} \cup \{110,111\}$. and the measurement will be 2-local.\\
\subsubsection*{Parity-like Born policy}
Notice that for the base case $|A|=2$, the contiguous-like Born policy loses expressivity since the measurement becomes 1-local. We can actually devise a more expressive assignment by considering a parity function, as illustrated in Figure \ref{fig: partitions base case}(b). The $2^n$ bitstrings in a $n$-qubit PQC can be considered assigning each of them by the parity of the bitstring (number of 1's). Thus, the policy is represented as:
\begin{equation}
    \pi(a|s,\theta) = \sum_{b \in \{0,1\}^n}^{\oplus b = a} \bra{\psi(s,\theta)} b \rangle \langle b \ket{\psi(s,\theta)}
\end{equation}
Such an assignment constitutes a global measurement and the authors of \cite{meyer_quantum_2023} showed that it corresponds to the assignment that maximizes the extracted information. Notice that instead of the Pauli-Z measurement on every qubit, one could instead measure either a single-qubit or an ancilla, provided a CNOT cascade prior to the measurement, as highlighted in \cite{meyer_quantum_2023}. For $|A|>2$, the authors designed a partitioning based on a recursive parity function which they conjecture to be optimal in the sense of extracted information and globality. Let $m=\log |A|$ be the number of recursive calls and $\boldsymbol{b}$ be a $n$-bit bitstring measured through sampling a PQC. Then, the partition can be defined recursively as,
\begin{equation}
    \mathcal{C}_{[a]_2}^{(m)}=\left\{\boldsymbol{b} \mid \bigoplus_{i=m}^{n-1} b_i=a_0 \wedge \boldsymbol{b} \in \mathcal{C}_{a_m \cdots a_2\left(a_1 \oplus a_0\right)}^{(m-1)}\right\}
\end{equation}
where $[a]_2 = a_m \dots a_0$ is the binary expansion of action $a$. Since for computing the parity, each of the $n$ bits is necessary, a parity-based policy will be composed of a global measurement (or $n$-local) for $|A|=2$ as base case. Thus, it will always be global independently of the number of actions.\\
\subsubsection*{Action-projector-like Born policy}
There can also be partitions that do not take into account every eigenstate of a given observable. For instance, for the base case $|A|=2$, we could assign the all-zero state to action $a_0$, $\langle P_{a_0} \rangle_{s,\theta} = |\langle 0 | \psi(s,\theta)\rangle|^2$ and the all-ones state to action $a_1$, $\langle P_{a_1} \rangle_{s,\theta} = |\langle 1 | \psi(s,\theta)\rangle|^2$, and discard all other basis states, as illustrated in Figure \ref{fig: partitions base case}(c). In this case, the probability would need to be further normalized:
\begin{equation}
    \pi(a|s,\theta) = \frac{\langle P_{a} \rangle_{s,\theta}}{\sum_{a' \in A} \langle P_{a'} \rangle_{s,\theta}}
\end{equation}
For $|A|=2^n$ it makes sense to assign each eigenstate to an action. In such case, the measurement would be $n$-local.\\
Table \ref{tab: born policies} summarizes the locality of the measurement for the different partitions considered in this work. The locality is expressed as a function of $|A|$.\\

\begin{table}[h!]
    \centering
    \begin{tabular}{|l|l|}
    \hline
    \textbf{Born policy} & \textbf{Locality of measurement} \\ \hline
    Contiguous-like & $\log |A|$ - local \\ \hline
    Parity-like & $n$-local \\ \hline
    Action-projector-like & $n$-local \\ \hline
    \end{tabular}
    \caption{Summary of the locality of the measurement for the different partitions considered in this work.}
    \label{tab: born policies}
 \end{table}

\subsection{Softmax policy}
\begin{restatable}{defi}{sm}
    \label{def: softmax policy}
    Let $s \in \mathcal{S}$ be a state embedded in an $n$-qubit parameterized quantum state, $\ket{\psi(s,\theta)}$, where $\theta \in \mathbb{R}^k$. Let $O_a$ be an arbitrary observable composed by the sum of $m$ local/global terms $O_a  = \sum_{i=0}^{m-1} \langle O_i \rangle$ representing the numerical preference of action $a \in \mathcal{A}$ and $\beta$ an hyperparameter. The probability associated to a given action $a$ for a softmax policy is given by: 
    \begin{equation}
        \pi(a | s , \theta) = \frac{e^{\beta\langle O_a \rangle_{s,\theta}}}{\sum_{a'} e^{ \beta\langle O_{a'} \rangle_{s,\theta}}}
    \end{equation}
\end{restatable}
where $\beta$ is often referred as the inverse temperature hyperparameter that is responsible for the control of the policies greediness. That is, the softmax policy allows for greater control compared to the Born policy , since $\beta$ can control the degree in which we select what we think to be the best action or explore other actions. The higher the $\beta$ the more greedy the policy is \cite{jerbi_variational_2021}.
\subsection{Gradient estimation}
The policy gradient (Equation \ref{eq: policy gradient baseline}) is in its essence classical with the exception of the log policy gradient in which the gradient w.r.t the PQC must be computed. In that regard, the log policy gradient must be expressed as the gradient of the expectation value of an observable and the parameter-shift rule \cite{schuldEvaluatingAnalyticGradients2019} can be applied to compute the gradient using quantum hardware. Let $\langle O \rangle_{\theta}$ be the parameterized expectation value of the observable $O$. The parameter-shift rule is a hardware-friendly technique to compute the partial derivative of $\langle O \rangle_{\theta}$ w.r.t $\theta$. Explicitly, it states the equality

\begin{equation}
    \frac{\partial \langle O \rangle_{\theta}}{\partial \theta_l} = \frac{1}{2 \sin \alpha} \bigl[ \langle O \rangle_{\theta+\alpha e_l} - \langle O \rangle_{\theta-\alpha e_l}\bigr]
    \label{eq: parameter-shift rule}
\end{equation}
where $e_l$ indicates that the parameter $\theta_l$ is being shifted by $\alpha$. The partial derivative can be obtained using two expectation value estimates, each requiring a number of quantum circuit evaluations. Thus, for $\theta \in \mathbb{R}^k$, the gradient can be estimated using $2k$ total quantum circuit evaluations. The gradient accuracy is maximized at $\alpha = \frac{\pi}{4}$, since $\frac{1}{\sin \alpha}$ is minimized at this point. For arbitrary functions of expectation values like the log policy gradient, the gradient can be obtained via the standard chain rule.\\
For the Born policy the chain rule gives the following expression for the log policy gradient partial derivatives
\begin{equation}
    \partial_{\theta_{l}} \log \pi(a|s,\theta) = \partial_{\theta_{l}} \log \langle P_a \rangle_{s,\theta} = \frac{\partial_{\theta_l} \langle P_a \rangle_{s,\theta}}{\langle P_a \rangle_{s,\theta}}
    \label{eq: log policy gradient expression born}
\end{equation}
which results clearly in a unbounded gradient expression. 
\begin{comment}
For the Softmax policy, the log policy gradient can be expanded in terms of a really simple form of derivative of expectation values \cite{sequeira_policy_2023}:
\begin{equation}
    \partial_{\theta_{l}} \log \pi(a|s,\theta) = \partial_{\theta_{l}} \langle O_a \rangle_{s,\theta} - \sum_{a'} \partial_{\theta_{l}} \langle O_{a'} \rangle_{s,\theta} \pi(a'|s,\theta)
    \label{eq: log policy gradient expression softmax}
\end{equation}
For observables composed by $m$ terms as proposed in Definition \ref{def: softmax policy}, the gradient is bounded $\partial_{\theta_{l}} \log \pi(a|s,\theta) \in [-m,m]$. 
\end{comment}
The full REINFORCE algorithm with PQC-based policies explored in this work is outlined in Algorithm \ref{alg: reinforce}.

\begin{algorithm}[!ht]
    \label{alg: reinforce}
    \DontPrintSemicolon
      
    \KwInput{PQC-based policy $\pi_{\theta}$ with $\theta \in \mathbb{R}^k$. Learning rate $\eta$ and horizon $T$}
    \KwOutput{Updated parameters $\theta^{*}$}
    
    \tcc{Loop until the stopping condition is met}
    \While{True}
    {
        \tcc{Generate trajectories following the policy $\pi_{\theta}$}
        \For{$i=0 \ldots N-1$}
        {
            Generate $\tau_i = \{(s_0,a_0,r_0),\dots ,(s_{T-1},a_{T-1},r_{T-1})\}$ under PQC-based policy\;
        }
        
        Compute gradient with baseline, $\nabla_{\theta} J(\theta)$ as in Equation \eqref{eq: policy gradient baseline} using the parameter-shift rule (Equation \ref{eq: parameter-shift rule})    \;
        \tcc{Update parameters via gradient ascent}
        $\theta = \theta + \eta \nabla_{\theta} J(\theta)$\;
    }
    
    \caption{PQC-based REINFORCE with Baseline}
\end{algorithm}

\section{Trainability issues in Born policies}\label{sec: pg variance}

This section presents new findings that form the cornerstone of this study, addressing the trainability of the quantum policy gradient algorithm as outlined in Algorithm \ref{alg: reinforce}. We focus on Contiguous and Parity-like PQC-based Born policies defined in Definition \ref{def: born policy}. A key aspect of this investigation is the analysis of the variance of the log policy gradient for these policies, considering the impact of the number of qubits and actions, which subsequently influences the globality of associated observables, as detailed in Table \ref{tab: born policies}. The analysis proceeds as follows:
\begin{enumerate}
    \item \textit{Analysis of Product States} (Subsection \ref{subsec: product states}): We begin with an examination of product states as an instructive case, discussing the behavior and characteristics of the log policy gradient variance in this simplified scenario.
    \item \textit{Consideration of Entangled States} (Subsection \ref{subsec: entangled states}): We extend the analysis to include entangled states, comparing and contrasting the findings with those from the product states to highlight the effects of entanglement on trainability.
    \item \textit{Unified Variance Analysis} (Subsection \ref{subsec: var as a function of |A|}): We conduct a unified analysis of variance as a function of the number of actions, providing a comprehensive overview of how the variance scales with an increasing number of actions and its implications for the trainability of PQC-based policies.
\end{enumerate}
By systematically analyzing these cases, we aim to provide a thorough understanding of the factors influencing the trainability of quantum policy gradient algorithms and offer insights into optimizing PQC-based policies for practical applications. Since the variance of the log policy partial derivative is desired, we start with a simplification of the REINFORCE policy gradient objective, expressed in Equation \eqref{eq: policy gradient baseline}, to an expression that depends only on the variance of the policy. This approach allows for an accurate study of trainability as a function of different PQC-based policies. In the following, we consider the trivial upper bound in terms of relevant quantities in reinforcement learning to rephrase the variance expression as a function of the policy.

\begin{restatable}{lem}{lemmaqpgvarianceexpectation}
    \label{lemma: qpg variance expectation}
    Let $\pi(a|s,\theta)$ be a $n$-qubit PQC-based policy with $\theta \in \mathbb{R}^k$. Let $T$ be the trajectories horizon, $R_{\text{max}}$ be the maximum reward and $\gamma$ the trajectories discount factor. Then, the policy gradient variance w.r.t variational parameters $\theta$ is upper bounded by
    
    \begin{equation}
        \mathbb{V}_{\theta}\biggl[\partial_{\theta} v_{\pi}(s) \biggr] \leq \frac{R_{max}^2 T^4}{(1-\gamma)^4} \mathbb{V}_{\theta} \bigl[\partial_{\theta} \log \pi(a|s,\theta) \bigr]
    \end{equation}
    
\end{restatable}    
\begin{proof}

\begin{alignLetter}
 \mathbb{V}_{\theta} \left[\partial_{\theta} v_{\pi}(s) \right] &= \mathbb{V}_{\theta}\left[\frac{1}{N} \sum_{i=0}^{N-1} \sum_{t=0}^{T-1} G_t(\tau_i) \partial_{\theta} \log \pi(a_t^i | s_t^i , \theta) \right] \nonumber\\
 &= \frac{1}{N^2}\mathbb{V}_{\theta}\left[\sum_{i=0}^{N-1} \sum_{t=0}^{T-1} G_t(\tau_i) \partial_{\theta} \log \pi(a_t^i | s_t^i , \theta) \right] \nonumber\\
 &\leq \frac{1}{N^2} \left(\sum_{i=0}^{N-1} \sum_{t=0}^{T-1} \sqrt{G_t^2 \mathbb{V}_{\theta} \left[\partial_{\theta} \log \pi(a_t^i | s_t^i , \theta)\right]}\right)^2\\
 &= \frac{G_t^2}{N^2} \left(\sum_{i=0}^{N-1} \sum_{t=0}^{T-1} \sqrt{\mathbb{V}_{\theta} \left[\partial_{\theta} \log \pi(a_t^i | s_t^i , \theta)\right]}\right)^2\\
 &\leq G_t^2 T^2 \mathbb{V}_{\theta} \left[\partial_{\theta} \log \pi(a | s , \theta)\right]\\
 &= \frac{R_{max}^2 T^4}{(1-\gamma)^4} \mathbb{V}_{\theta} \left[\partial_{\theta} \log \pi(a | s , \theta)\right]
\end{alignLetter}
\noindent
where: (A) Follows from the variance of the sum of random variables $\bigl(\mathbb{V} \left[\sum_{i} X_i \right] \leq \left( \sum_i \sqrt{\mathbb{V}[X_i]} \right)^2 \bigr)$. (B) Follows from variance of a constant $a$ times a random variable $X$ $(\mathbb{V}[aX] = a^2 \mathbb{V}[X])$. (C) Considers the upper bound on $N$ and $T$. (D) Considers the trivial upper bound on the return (see Appendix \ref{appendix: upper bound return}), following the independence of $\theta$.
\noindent
\end{proof}
Lemma \ref{lemma: qpg variance expectation} indicates that the variance of the log policy objective function increases with relevant quantities in RL. Specifically, the variance increases with the reachable maximum reward, the horizon and the discount factor. At this stage trainability of PQC-based agents can be evaluated through the scaling of the variance of the log policy gradient $\mathbb{V}_{\theta} \left[\partial_{\theta} \log \pi(a | s , \theta) \right]$ as a function on the number of qubits. To that end, let us start by analyzing the behavior of the gradients in the context of product states and build from there towards general entangled quantum states.\\

\subsection{The instructive case of product states}\label{subsec: product states}

    We begin by examining the straightforward scenario of a product state. In \cite{cerezo_cost_2021}, the authors explored a simple $n$-qubit parameterized model described by the unitary $V(\theta) = \bigotimes_{i=0}^{n-1} e^{-i \theta_i \sigma_x}$. They focused on the global observable $O_G = 1 - |0\rangle \langle 0|$ to prepare the all-zero state. Although this PQC corresponds to a single layer of parameterized Pauli rotations forming a separable state, it was shown to suffer from BPs. The global observable results in a cost function $C_G(\theta) = 1 - \prod_{i=0}^{n-1}\cos^2(\theta_i)$, whose variance decays exponentially with the number of qubits due to its global nature. The authors then suggested the local observable composed by individual qubit contributions $O_L = 1 - \sum_{j=0}^{n-1}\ket{0}\bra{0}_j \otimes \mathbb{I}_{\bar{j}}$ with cost function $C_L(\theta) = 1 - \frac{1}{n}\sum_{i=0}^{n-1}cos(\theta)^2$. This modification ensures that the variance of the cost function decays polynomially with the number of qubits, thus avoiding BPs. Such finding emphasizes the critical role of a well-crafted cost function.\\
    \begin{figure}[!htb]
        \centering
        \begin{minipage}[t]{\textwidth}
        \centering
        \includegraphics[width=\textwidth]{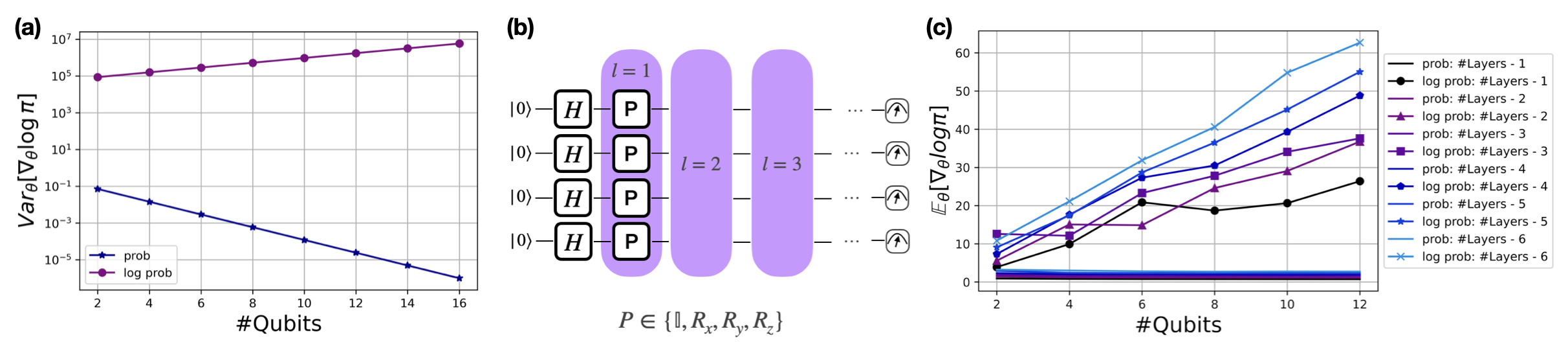}
        \end{minipage}%
        \hfill
        \caption[]{Variance and expectation value of the gradient of log-probability cost-function. \textbf{(a)} Variance for the all-zero state with parameterized state of individual qubit y-rotations. \textbf{(b)} Random product state composed of Pauli rotations sampled uniformly at random. \textbf{(c)} Expectation value for randomly sampled projectors in the random product state in (b).}
        \label{fig: instructive_product_states}
    \end{figure}
    In the broader context of machine learning, and policy gradients specifically, the log-likelihood is often preferred over direct probability as considered before. Such cost-function leads to different behavior. For an arbitrary product state $|\psi\rangle$, the probability of the all-zero state is given by:
    \begin{equation}
        |\braket{0}{\psi}|^2 = \prod_{i=0}^{n-1} |\braket{0_i}{\psi}|^2
    \end{equation}
    The decomposition into individual qubit contributions enables a product state to avoid BPs since the log likelihood cost-function separates the product into a sum of individual qubit contributions. To apply this reasoning to the REINFORCE cost function in RL, where the focus is on the log policy gradient, consider a Born policy with $|A|=2^n$ and a global projector $\ket{a}\bra{a}$ for action $a$. The policy is expressed as $\pi(a|s,\theta) = |\braket{a}{\psi(s,\theta)}|^
    2$. If the parameterized state is a product state, the probability can be decomposed into individual qubit contributions as follows:
    
    \begin{equation}
    \pi(a|s,\theta) = |\braket{a}{\psi(s,\theta)}|^2 = \prod_{i=0}^{n-1} |\braket{a_i}{\psi(s,\theta)}|^2
    \end{equation}
    where $a_i$ represents the individual qubit projector $\ket{a_i}\bra{a_i} \otimes \mathbb{I}_{\hat{i}}$ on the $i^{\text{th}}$ qubit, applying the identity operation to the other qubits. Considering the variance of the log policy gradient:

    \begin{alignLetter}
        \mathbb{V}_{\theta}\biggl[\partial_{\theta} \log \pi(a|s,\theta) \biggr] &= \mathbb{V}_{\theta}\biggl[\partial_{\theta} \log \prod_{i=0}^{n-1} |\braket{a_i}{\psi(s,\theta)}|^2 \biggr] \nonumber \\
        &= \mathbb{V}_{\theta}\biggl[\sum_{i=0}^{n-1} \partial_{\theta} \log |\braket{a_i}{\psi(s,\theta)}|^2 \biggr] \nonumber \\
        &= \sum_{i=0}^{n-1} \mathbb{V}_{\theta}\biggl[\partial_{\theta} \log |\braket{a_i}{\psi(s,\theta)}|^2 \biggr]
    \end{alignLetter}
    where (A) follows from the linearity and independence of the observables \cite{uvarovBarrenPlateausCost2021}. The variance of the log policy gradient becomes the sum of the variances of the log probabilities of each individual qubit. Notice that since we have a product state, the partial derivative would in fact not depend on the number of qubits, provided that different parameters are part of the circuit. Only in the scenarion where the parameters are shared across qubits the partial derivative will sum those terms and increase with the number of qubits, as illustrated in Figure \ref{fig: instructive_product_states}(a). Such behavior is propelled by the nature of a product state, where the probability of each individual qubit is independent of the other qubits. In Figure \ref{fig: instructive_product_states}(c) the expectation value of the partial derivative of log probability of the all-zero state is illustrated for the random product state illustrated in Figure \ref{fig: instructive_product_states}(b) where the parameters are shared per layer. That is $\theta_{i,l} = \theta_{l}$for all number of layers $l$. Indeed, the variance increases with both the number of qubits and layers, as expected.

\subsection{Generalized behavior for entangled states}\label{subsec: entangled states}
In this subsection, we analyze the variance of the log-probability for entangled states. In particular, we focus on the extreme case where $|A|=2^n$, involving global projectors similar to the product states discussed in Subsection \ref{subsec: product states}. It is known that such measurements are susceptible to barren plateaus (BPs) \cite{cerezo_cost_2021} since the probability of each basis state in this scenario depends on a subset of qubits characterized by the entangled state, derived from an $n$-qubit global projector, assuming a PQC constituted by local two-design parameterized blocks. In Figure \ref{fig: generalized_entangled_variance}\textbf{(d)} the variance of the log probability is illustrated as a function of the number of qubits for three distinct entangled quantum states: 1) Simplified 2-design ansatz illustrated in Figure \ref{fig: generalized_entangled_variance}\textbf{(a)}. 2) Strongly entangling layers, depicted in Figure \ref{fig: generalized_entangled_variance}\textbf{(b)}. 3) State generated from Pauli rotations sampled uniformly at random followed by randomly selected CZ gates, as illustrated in \ref{fig: generalized_entangled_variance}\textbf{(c)}. $n$ layers of the blocks shown in their respective figures are employed. Moreover, projectors were sampled uniformly at random from the set of $2^n$ available ones and the variance illustrated for an average of a thousand experiments.

\begin{figure}[!htb]
    \centering
    \begin{minipage}[t]{\textwidth}
    \centering
    \includegraphics[width=\textwidth]{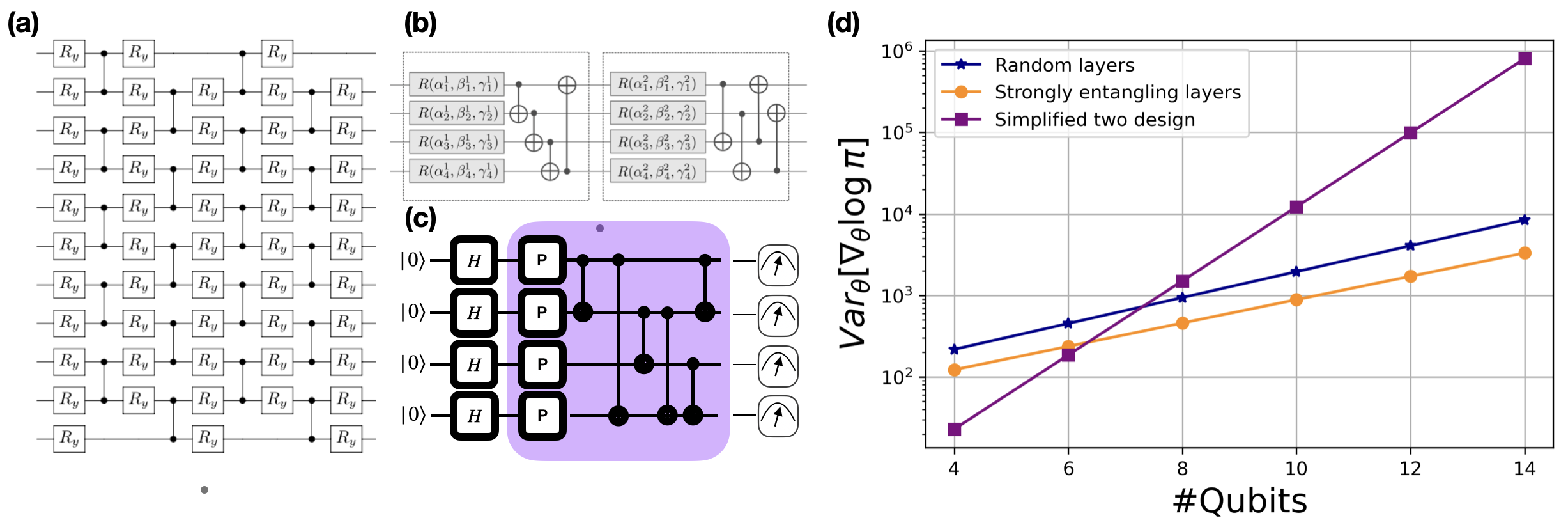}
    \end{minipage}%
    \hfill
    \caption[]{variance of the log policy gradient for three distinct entangled states. \textbf{(a)} Simplified two design. \textbf{(b)} Strongly entangling layers. \textbf{(c)} Random states composed of Pauli rotations sampled uniformly at random followed by randomly selected CZ gates. \textbf{(d)} Variance as a function of the number of qubits for the circuits \textbf{(a)}-\textbf{(c)}.}
    \label{fig: generalized_entangled_variance}
\end{figure}
From Figure \ref{fig: generalized_entangled_variance}\textbf{(d)}, it is evident that in each experiment, the variance of the log-probability increases with the number of qubits when global projectors are considered. This behavior is akin to that observed in product states. However, the variance reaches extremely high levels as a function of $n$, indicating that although these circuits avoid BP, they are prone to exploding gradients. This phenomenon arises because the probabilities diminish exponentially with an increase in the number of qubits, leading to two major issues: 1) The log-probability gradient becomes exponentially large due to the vanishing probabilities. 2) An exponentially large number of quantum circuit executions is required to accurately estimate both the probability and its gradient. As the number of qubits grows, measuring the eigenstate of interest becomes increasingly challenging due to the exponentially concentrated probabilities \cite{rudolphTrainabilityBarriersOpportunities2023}. 
However, recall that in the context of RL, we will need to do a partitioning of possibly all $2^n$ basis states into the set of available actions $|A|$. Thus, the previous observation is no longer true once the number of actions is $|A| \in \mathcal{O}(\text{poly}(n))$ since the probabilities will no longer be exponentially small . In such cases, a trainable region could be created depending on the locality of the projector, which in turn is heavily influenced by the type of Born policy implemented. In the following subsection, we examine the variance of the cost function for different Born policies as a function on the number of actions $|A|$.

\subsection{Variance as a function of $|A|$}\label{subsec: var as a function of |A|}
Let us start with an analytical upper bound for the variance of the log likelihood cost function partial derivative, presented in Lemma \ref{lemma: upper bound variance born}. Let $f(\pi_\theta) = \log \pi(a|s,\theta)$ for simplicity.

\begin{restatable}{lem}{bornpolicygradientlemma}
    \label{lemma: upper bound variance born}
    Consider a $n$-qubit Born policy $\pi(a|s,\theta)$ as in Definition \ref{def: born policy} with $|A|$ actions. Then, the upper bound for the variance of the log policy gradient is given by
    \begin{equation}
        \mathbb{V}_\theta\biggl[\partial_{\theta} \log \pi(a|s,\theta)\biggr] \leq 2\left|\partial_{\pi_\theta} f\left(\pi_\theta\right)\right|_{\infty}^2\left[\mathbb{V}_\theta\left[\partial_\theta \pi_\theta\right]+\mathbb{E}_\theta\left[\partial_\theta \pi_\theta\right]^2\right]
    \end{equation}
\end{restatable}
\begin{proof}
\begin{alignLetter}
 \mathbb{V}_\theta\biggl[\partial_{\theta} \log \pi(a|s,\theta)\biggr] &= \mathbb{V}_\theta\biggl[ \partial_{\theta} f(\pi_\theta) \biggr] \nonumber\\
 &= \mathbb{V}_\theta\biggl[\partial_{\pi_\theta} f(\pi_\theta) \partial_{\theta} \pi_\theta \biggr]\\
 &\leq 2 \mathbb{V}_\theta\left[\partial_\theta \pi_\theta\right]\left|\partial_{\pi_\theta} f\left(\pi_\theta\right)\right|_{\infty}^2+2 \mathbb{E}_\theta\left[\partial_\theta \pi_\theta\right]^2 \mathbb{V}_\theta\left[\partial_{\pi_\theta} f \pi_\theta\right] \\
 &\leq 2 \mathbb{V}_\theta\left[\partial_\theta \pi_\theta\right]\left|\partial_{\pi_\theta} f\left(\pi_\theta\right)\right|_{\infty}^2+2 \mathbb{E}_\theta\left[\partial_\theta \pi_\theta\right]^2\left|\partial_{\pi_\theta} f\left(\pi_\theta\right)\right|_{\infty}^2 \\
 &=2\left|\partial_{\pi_\theta} f\left(\pi_\theta\right)\right|_{\infty}^2\left[\mathbb{V}_\theta\left[\partial_\theta \pi_\theta\right]+\mathbb{E}_\theta\left[\partial_\theta \pi_\theta\right]^2\right]
\end{alignLetter}
where (A) Follows from the chain rule; (B) Follows from variance of product of random variables $\mathbb{V} \left[ XY \right] \leq 2 \mathbb{V} \left[ X \right]|Y^2|_{\infty} + 2 \mathbb{E}\left[X\right]^2 \mathbb{V} \left[Y\right]$ as proposed in \cite[]{thanasilp_subtleties_2021}; (C) Upper bound on the variance; (D) Algebraic manipulation.
\end{proof}
The upper bound depends entirely on the total number of actions and the observable considered to estimate the policy. Let us assume a global projector on $n$-qubits and either parameterized blocks before/after parameter $\theta$ form a 1-design. That way, the average of the partial derivative ${E}_\theta\left[\partial_\theta \pi_\theta\right]=0$ \cite{cerezo_cost_2021}. 
If $|A| \in \mathcal{O}(\text{poly}(n))$, then w.l.g we can assume that $\pi_{\text{min}} \in [b,1]$ with $b \in \Omega(\frac{1}{\text{poly}(n)})$ \cite{thanasilp_subtleties_2021}. In the context of RL, the log policy gradient is only computed for sampled actions. Thus, the probability of an action cannot be zero in the gradient estimation phase. Nevertheless, it can be arbitrarily close to zero. To avoid such an issue, clipping is often considered in practice. Thus, assuming $b \in \Omega(\frac{1}{\text{poly}(n)})$ works as some sort of clipping of probabilities. In general, provided discrete action spaces, $|A| \ll 2^n$. Therefore it is only reasonable to assume the possible number of outcomes to be $|A|\sim \text{poly}(n)$. Thus, considering $b \in \Omega(\frac{1}{\text{poly(n)}})$ as proposed in \cite{thanasilp_subtleties_2021} does not pose any issues under these conditions. However, this assumption breaks down in the most general case when $|A|=2^n$, associating each basis state projector in an $n$-qubit PQC to an action, as the probabilities become exponentially small with the number of qubits. Consequently, for action spaces larger than $\text{poly}(n)$, the number of quantum circuit executions required to accurately estimate the policy becomes impractical, eventually scaling exponentially with the number of qubits.
Given that the total number of features in an RL agent's state, $s_f$, is typically large and $s_f \gg |A|$ for a discrete action space, several qubits are often required to encode the state of the agent. If the standard angle encoding scheme is employed, as seen in most literature \cite{jerbi_quantum_2022,chen_variational_2020,sequeira_policy_2023,jerbi_variational_2021}, then $n \sim s_f$, which implies that $|A| \ll 2^n$, validating the $\text{poly}(n)$ clipping assumption. 
Therefore, provided $|A| \in \mathcal{O}(\text{poly}(n))$ the absolute value of the log policy gradient $\left|\partial_{\pi_\theta} f\left(\pi_\theta\right)\right|_{\infty}^2$ falls within $\mathcal{O}(\text{poly}(n))$, with $b \in \Omega(\frac{1}{\text{poly}(n)})$.
If each parameterized block forms a local 2-design, the variance $\mathbb{V}_\theta\left[\partial_\theta \pi_\theta\right] \in \mathcal{O}(\frac{1}{\alpha^n})$ for $\alpha > 1$. Thus, in this setting the overall variance vanishes exponentially with the number of qubits.\\

However, outside of $\text{poly}(n)$ actions, the polynomial clipping assumption no longer holds. In such cases, the probability of a given action can be exponentially small, leading to a concentration with the number of qubits \cite{rudolphTrainabilityBarriersOpportunities2023}. Thus, for  $\beta > 1$, $\pi_{\text{min}} \in \Omega(\frac{1}{\beta^n})$ and $\left|\partial_{\pi_\theta} f\left(\pi_\theta\right)\right|_{\infty}^2  \in \mathcal{O}(\beta^n)$. Consequently, the variance $\mathbb{V}_\theta\biggl[\partial_{\theta} \log \pi(a|s,\theta)\biggr]$ scales as $\mathcal{O}\biggl( \biggl(\frac{\beta}{\alpha}\biggr)^n \biggr)$ and the upper bound becomes too lose since $\beta$ could actually be greater than $\alpha$ which implies that the variance increases with the number of qubits. 
Nonetheless, the behavior analyzed above corresponds exactly to the behavior observed with a parity-like Born policies, since this policy is always composed of a $n$-local measurement. Thus, for $|A| \in \text{poly}(n)$ the policy is prone to BPs. Beyond $\text{poly}(n)$ actions, the exponentially small probabilities cause a phase transition from BPs to exploding gradients. However, the need for an exponentially large number of shots to accurately estimate the policy renders it untrainable for a large number of qubits and actions.

    \label{lemma: lower bound contiguous-like}
    In stark contrast, the base case $|A|=2$ in a contiguous-like Born policy hinges on single qubit measurements which implies right from the start, a markedly different trainability profile compared to the parity-like policy. It is expected that the contiguous-like policy becomes harder to train with an increase in the number of actions since that translates directly in the globality of the employed observables. However, there exists a window of trainability that is contingent on maintaining a relatively small number of actions. Specifically, when the number of actions equals $|A|=n$, the contiguous-like policy utilizes $\log(|A|)$-local measurements. In practical terms, this equates to measuring at most $\log(n)$ adjacent qubits, a scenario that is typically local and known to circumvent BPs \cite{rudolphTrainabilityBarriersOpportunities2023}. The trainability window is described in Lemma \ref{lemma: lower bound contiguous-like}.

\begin{restatable}{lem}{bornpolicygradientlemma}
    \label{lemma: lower bound contiguous-like}
    Consider a $n$-qubit contiguous-like Born policy $\pi(a|s,\theta)$ with $|A|$ actions as in Definition \ref{def: born policy}. Then, if each block in the parameterized quantum circuit forms a local 2-design, the policy gradient variance is given by
    
\begin{equation}
    \mathbb{V}_{\theta}\biggl[\partial_{\theta} \log \pi(a|s,\theta) \biggr] \in \Omega \biggl(\frac{1}{\text{poly(n)}} \biggr)
\end{equation}
for $|A| \in \mathcal{O}(n)$ and depth $\mathcal{O}(\log(n))$. On the other hand, the policy gradient variance scales as
\begin{equation}
    \mathbb{V}_{\theta}\biggl[\partial_{\theta} \log \pi(a|s,\theta) \biggr] \in \Omega \biggl(2^{-\text{poly}(\log(n))} \biggr)
\end{equation}
for $|A| \in \mathcal{O}(n)$ and depth $\mathcal{O}(\text{poly}\log(n))$.
\end{restatable}
The detailed proof of this lemma can be found in Appendix \ref{appendix: lower bound contiguous-like}. It provides a lower bound for the variance of the policy gradient for contiguous-like Born policies, under the condition that each parameterized block in the circuit constitutes a local 2-design. For $|A| \in \mathcal{O}(n)$, the variance declines at most polynomially with the number of qubits, and the number of quantum circuit executions needed for accurate policy and gradient estimation does not grow exponentially, thanks to $\log(n)$-local measurement. Hence, the policy remains trainable up to a depth of $\mathcal{O}(\log(n))$. Thus, in general, for $|A|$ within $\mathcal{O}(\text{poly}(n))$, the variance diminishes more rapidly than polynomially but less so than exponentially with the number of qubits, due to the observables' locality being greater than $\log(n)$ but less than $n$.\\

\subsection{Analysis of the Fisher Information spectrum}\label{subsec: analysis of the Fisher Information spectrum}

In the realm of computational learning theory, the Fisher Information Matrix (FIM) is instrumental for evaluating the information garnered from a parameterized statistical model. In RL, the FIM must account for states sampled from the distribution generated under the policy, our parameterized model. Defining $\pi(a|s,\theta)$ as the parameterized policy and $d_{s}^{\pi}$ as the distribution of states under the policy, the FIM is formulated as the expected value of the outer product of the gradient of the log-likelihood function:

\begin{equation}
    \mathcal{I}(\theta) = \mathbb{E}_{s \sim d_{s}^{\pi}}\mathbb{E}_{a \sim \pi(.|s,\theta)} \biggl[ \nabla_{\theta} \log \pi(a|s,\theta) \nabla_{\theta} \log \pi(a|s,\theta)^T \biggr]   
\end{equation}
The FIM efficiently indicates how changes in parameters impact the model's output. Notably, the FIM spectrum is crucial for characterizing the barren plateau (BP) phenomenon in PQC-based statistical models employing log-likelihood loss functions \cite{abbasPowerQuantumNeural2021} which in turn can also be used to characterize BPs in the RL paradigm. However, in RL, where the loss function is influenced by cumulative rewards, the FIM does not consider this weighting. Still, assuming non-zero rewards, the FIM spectrum remains a valuable tool for BP characterization. In a BP, the eigenvalues on the FIM will be exponentially concentrated around zero \cite{abbasPowerQuantumNeural2021}. The expected value of a diagonal entry $k$ of the FIM can be written as,

\begin{alignLetter}
    \mathbb{E}_{\theta}\biggl[\mathcal{I}_{kk}(\theta)\biggr] &= \mathbb{E}_{\theta} \biggl[ \biggl( \partial_{\theta_k} \log \pi(a|s,\theta) \biggr)^2 \biggr]\nonumber \\
    &= \mathbb{V}_{\theta} \biggl[\partial_{\theta_k} \log \pi(a|s,\theta) \biggr] + \biggl(\mathbb{E}_{\theta} \biggl[\partial_{\theta_k} \log \pi(a|s,\theta) \biggr]\biggr)^2
\end{alignLetter}
where (A) is obtained from the definition of the variance. As a result, the FIM's diagonal entry can be lower bounded by the variance of the log likelihood:

\begin{equation}
    \mathbb{E}_{\theta}\biggl[\mathcal{I}_{kk}(\theta)\biggr] \geq \mathbb{V}_{\theta} \biggl[\partial_{\theta_k} \log \pi(a|s,\theta) \biggr]
    \label{eq: lower bound FIM entry}
\end{equation}
This implies that the trace of the FIM, being the sum of its eigenvalues, has a lower bound as follows: $\text{Tr}\biggr[\mathbb{E}_{\theta}[\mathcal{I}(\theta)]\biggr] \geq \sum_{k=0}^{K-1} \mathbb{V}{\theta} \biggl[\partial{\theta_k} \log \pi(a|s,\theta) \biggr]$ for $\theta \in \mathbb{R}^K$. Therefore, the lower bound presented in Lemma \ref{lemma: lower bound contiguous-like} can be directly used to identify the conditions under which the FIM spectrum indicates a BP for PQC-based policies. In a BP scenario, each entry of the FIM will vanish exponentially with the number of qubits, necessitating an exponentially increasing number of measurements to estimate each entry of the FIM accurately. Based on Lemma \ref{lemma: lower bound contiguous-like}, for a Contiguous-like Born policy with $|A| \in \mathcal{O}(\text{poly}(n))$, the variance vanishes at most polylogarithmically with the number of qubits, affecting also the eigenvalues of the FIM as indicated by Equation \eqref{eq: lower bound FIM entry}. Consequently, the FIM spectrum in such cases fails to indicate a BP. Conversely, for a Parity-like Born policy with $|A| \in \mathcal{O}(\text{poly}(n))$, the variance shrinks exponentially with the number of qubits, as do the eigenvalues of the FIM highlighting a BP. 
In situations where the number of actions exceeds $\text{poly}(n)$, not only the number of required measurements for accurate policy estimation are prohibitively large, but also the probabilities associated with actions remain exponentially small. This implies that, despite avoiding BPs, these scenarios are more likely to encounter exploding gradients rather than BPs, reflected in increasing FIM entries and a less concentrated spectrum around zero. Therefore, it is crucial to note that in cases where the number of actions is large, the FIM spectrum may not indicate BPs, but significant trainability issues can still arise from the necessity of a large number of measurements and the possibility of exploding gradients.

\subsection{Summary}

In this section, we characterized the trainability landscape of PQC-based policies. We demonstrated that under the realistic condition of a polynomial number of measurements, the variance of the log policy gradient for a contiguous-like Born policy decays at most polylogarithmically with the number of qubits, while for a parity-like Born policy, the variance vanishes exponentially with the number of qubits, provided the action space is also of polynomial size. Outside of the regime of polynomially sized action spaces, both policies becomes become untrainable since exponentially small probabilities associated with actions arise, which induce an exploding gradient on the log policy gradient objective function. Such phenomena was also captured and characterized through the inspection of the Fisher Information spectrum. A spectrum highly concentrated around zero indicates a BP. However, we highlight that in a regime where the number of actions is outside of $\mathcal{O}(\text{poly}(n))$, the FIM spectrum will not faithfully indicate a trainability problem. In such cases since probabilities are exponentially small, the FIM spectrum will be less concentrated around zero, moving away from BPs. Nevertheless, the gradient explodes at the same time as an exponentially large number of measurements is needed to accurately estimate both the policy and gradient.

\section{Numerical experiments}\label{sec: numerical experiments}

In this section, we conduct an in-depth evaluation of the trainability issues in quantum policy gradients, as posited in Lemma \ref{lemma: lower bound contiguous-like}, through experimental validation in two distinct scenarios:
\begin{itemize}
    \item \textit{Trainability issues using a simplified 2-design} - We resort to the simplified two-design ansatz, as described in \cite{cerezo_cost_2021} and depicted in Figure \ref{fig: generalized_entangled_variance}(a). This task involves examining the variance of the log likelihood partial derivative with respect to the type of policy and the range of available actions. Additionally, we investigate the FIM spectrum for both policies, focusing on how it varies with the number of actions.
    \item \textit{Multi armed bandits} - To assess the practical performance of both Born policy formulations in an RL context, we designed a synthetic multi-armed bandit environment. The objective here is to evaluate the policies' ability to discern the optimal arm through actions sampling.

\end{itemize}
For the first task, the chosen two-design ansatz, while not strictly a two-design, it has been previously demonstrated to encounter cost-function BPs \cite{cerezo_cost_2021}. This setup is particularly advantageous for simulation purposes, especially with a large number of qubits and increased depth. We opted for a depth of $\mathcal{O}(n^2)$ in our experiments.
Given the scarcity of large-scale RL problems amenable to efficient resolution via PQC-based policies, and more critically, those employing a sufficient number of qubits to investigate trainability in policy gradients, we selected the multi-armed bandit for the second task. This choice allows us to maintain the same objective function while affording the flexibility to adjust the total number of qubits, thus facilitating a thorough examination of trainability issues. It's noteworthy that all experiments were conducted using Pennylane's quantum simulator \cite{bergholmPennyLaneAutomaticDifferentiation2022}, with gradient estimations carried out via Parameter-shift rules \cite{schuldEvaluatingAnalyticGradients2019} and considering $\mathcal{O}(\text{poly}(n))$ number of measurements. For those interested in replicating our findings, our work is available on the GitHub repository at \href{https://github.com/andre-sequeira10/Trainability-issues-in-QPGs}{Trainability-issues-in-QPGs}.
\subsection{Trainability issues using a simplified 2-design}
In Section \ref{subsec: var as a function of |A|}, we meticulously examine the conditions under which trainability issues emerge in the training of learning agents using quantum policy gradients, focusing on the type of policy employed and the range of actions available. This assessment hinges on the globality of the observable, and for this investigation, we adopt the simplified two-design ansatz referenced in \cite{cerezo_cost_2021}. To align closely with a RL setting, we incorporate the initial rotation set from the ansatz, shown in purple in Figure \ref{fig: S2D encoding}, to encode the agent's state. In this artificial setup, each of the agent's $n$ state features is encoded using $n$ qubits through angle-encoding. Both the agent's state, $s$, and the trainable parameters, $\theta$, are uniformly sampled from the interval $\{s,\theta\} \sim U(-\pi,\pi)$. 

\begin{figure}[!htb]
    \centering
    \begin{minipage}[t]{\textwidth}
    \centering
    \includegraphics[height=0.20\textheight,width=0.2\linewidth]{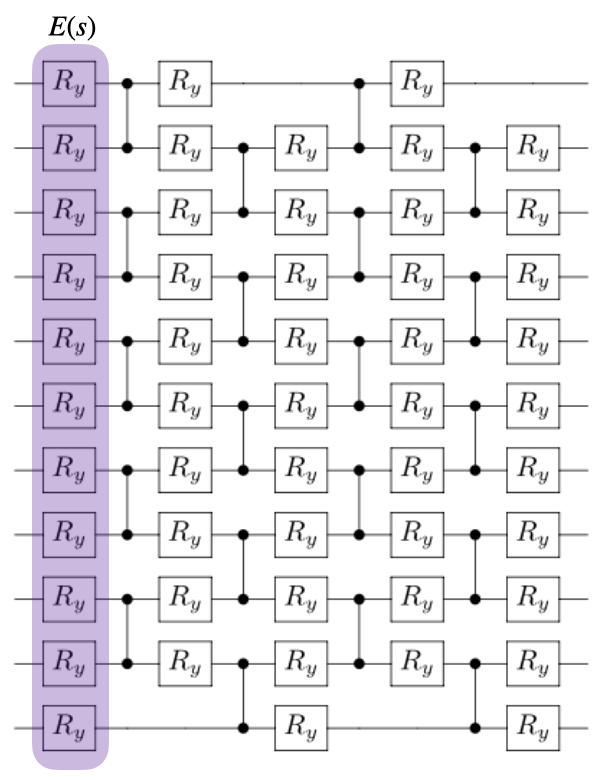}
    \end{minipage}%
    \hfill
    \caption[]{Simplified two-design ansatz with the first set of rotations in purple representing the unitary $E(s)$ responsible for the encoding of the state of the agent, $s$.}
    \label{fig: S2D encoding}
\end{figure}
Mimicking an RL environment, we execute the Parameterized Quantum Circuit (PQC) a polynomial number of times to construct the policy from the measurements. An action is then sampled from the resulting policy distribution, and the probability of the selected action is used to estimate the variance of the log likelihood's partial derivatives. This variance is evaluated as a function of the number of actions and qubits, with the number of qubits, $n$, considered within the finite set $n = \{4,6,8,10,12,14\}$. The action set, $A$, varies as a function of the number of qubits, defined as $A = \{ 2^i | i \in \mathbb{Z}\ ,1 \leq i \leq n\}$.
Considering the assumption that the minimum probability is bounded polynomially with $n$, a scenario typical in RL where actions are sampled from the policy's probability distribution, the probability cannot be zero but may approach it closely. For this reason, $\pi_{\text{min}} \in \Omega(\frac{1}{\text{poly}(n)})$ is assumed, provided the total number of actions is also polynomially large ($|A| \in \mathcal{O}(\text{poly}(n))$). This assumption acts as a practical clipping of probabilities, as often employed in RL contexts \cite{thanasilp_subtleties_2021}. However, as the number of qubits increases, probabilities become exponentially small, requiring an analysis of the log-likelihood cost function's variance under both clipped and unclipped probability scenarios. With probability clipping, $\pi_{\text{min}}$ is set to be within $\Omega(\frac{1}{n^2})$. Each experiment is repeated and averaged over 2000 iterations with a randomly generated parameter set. 
Regarding the FIM spectrum analysis, the same PQC is used under similar conditions. The only difference is in the number of executions, reduced to 10 to manage the computational cost of computing the FIM.
\subsubsection*{Contiguous-like Born policy}
Firstly, we explore the variance of the log-likelihood cost function for the contiguous-like Born policy. Figures \ref{fig: variances_contiguous-like_clamp} and \ref{fig: variances_contiguous-like_no_clamp} display this variance as a function of the number of actions, considering both clipped and unclipped probabilities. 
\begin{figure}[!htb]
    \centering
    \begin{minipage}[t]{\textwidth}
    \centering
    \includegraphics[height=0.20\textheight,width=\linewidth]{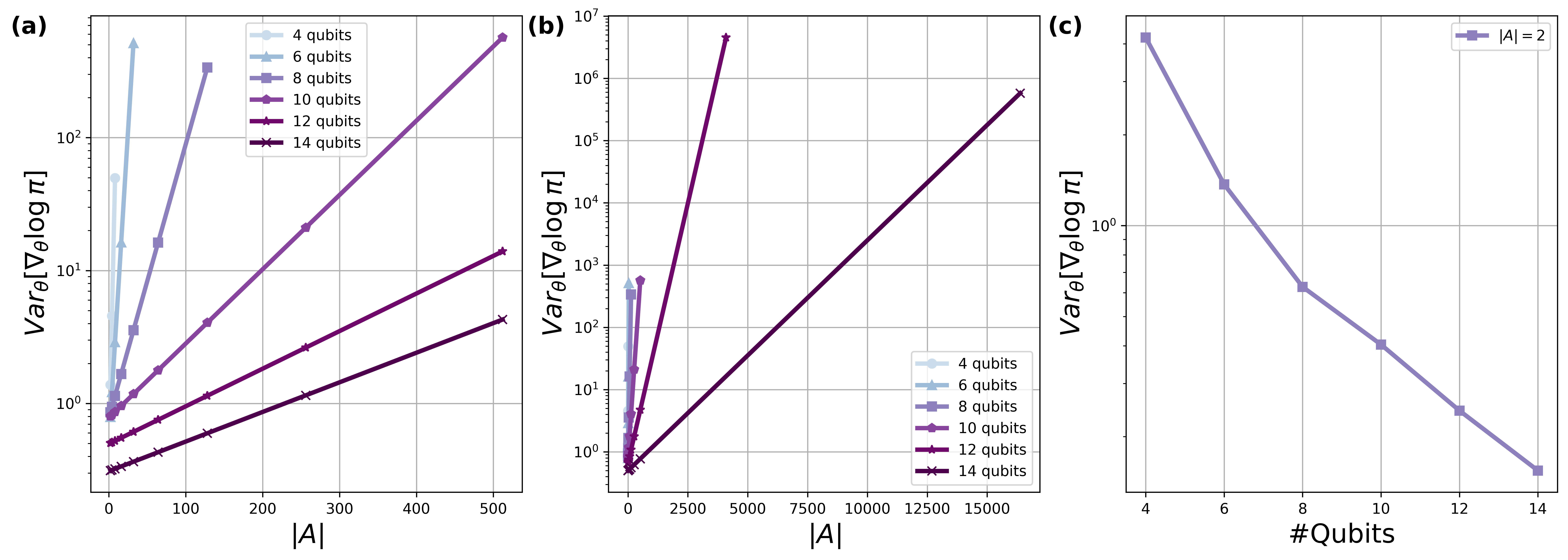}
    \end{minipage}%
    \hfill
    \caption[]{variance of the log policy gradient for contiguous-like Born policies: (a) and (b) as a function of $|A|$ and (c) semi-logarithmic plot for varying number of qubits.}
    \label{fig: variances_contiguous-like_no_clamp}
\end{figure}
In the unclipped probability scenario (Figure \ref{fig: variances_contiguous-like_no_clamp}(a)), the variance increases with the number of actions due to the $\log(|A|)$-local nature of the contiguous-like Born policy. However, despite the local nature of the measurements, as we increase the number of actions, probabilities diminish, leading to an escalation in the gradient. Figure \ref{fig: variances_contiguous-like_no_clamp}(c) reveals a semi-logarithmic plot for the variance illustrating a polynomial decay in variance with the number of qubits, confirming Lemma \ref{lemma: lower bound contiguous-like}'s predictions. Additionally, Figure \ref{fig: variances_contiguous-like_no_clamp}(b) demonstrates that the variance surges dramatically with an increase in actions, signaling the onset of exploding gradient phenomena. With a large number of qubits and actions, the gradients become unmanageable since the probabilities become exponentially small and thus a polynomial number of measurements will not be sufficient to perform learning. For that matter, consider a polynomial clipping of probabilities, illustrated in Figure \ref{fig: variances_contiguous-like_clamp}.
\begin{figure}[!htb]
    \centering
    \begin{minipage}[t]{\textwidth}
    \centering
    \includegraphics[height=0.20\textheight,width=\linewidth]{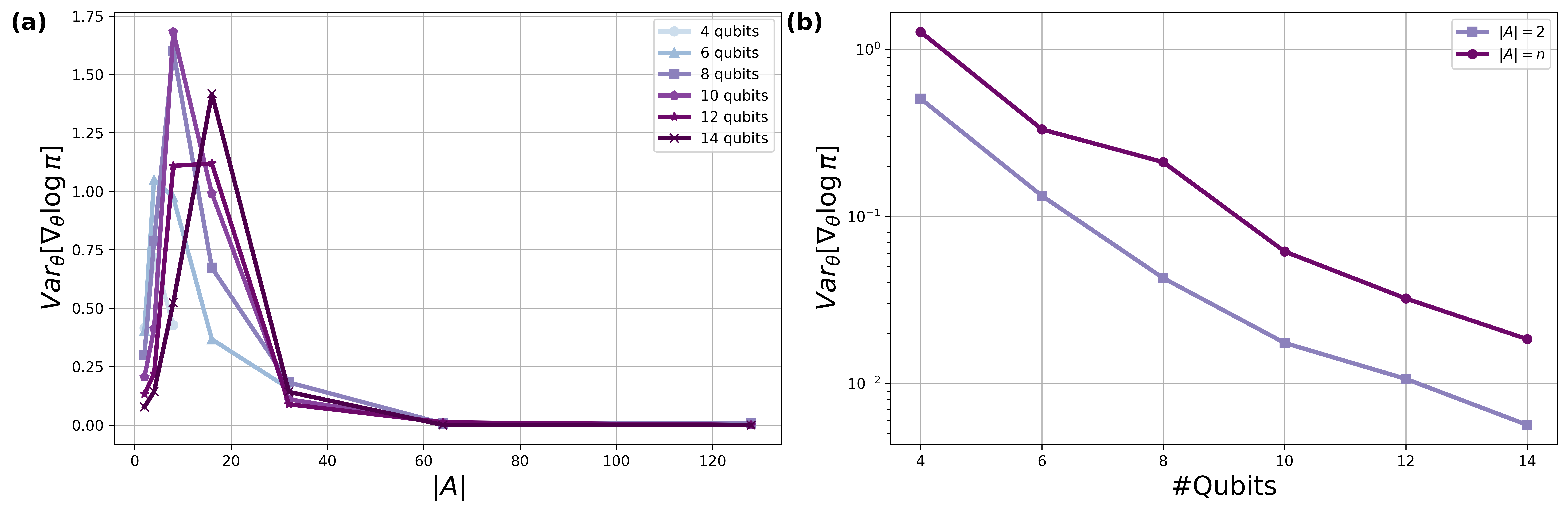}
    \end{minipage}%
    \hfill
    \caption[]{variance of the log policy gradient for contiguous-like Born policies as (a) a function of $|A|$ and (b) for varying number of qubits with polynomial clipping.}
    \label{fig: variances_contiguous-like_clamp}
\end{figure}
In this setting, it can be observed from Figure \ref{fig: variances_contiguous-like_clamp}(a) a different behavior for the scaling of the variance as a function of the number of actions, compared to the originally non-clipped probabilities. Firstly, for a small number of actions, the variance increases in the same way as before. However, at some point the variance starts decreasing heavily. Such behavior can be explained from the polynomial clipping considered. Recall that in this experimental setup $\pi_{\text{min}} \in \Omega(\frac{1}{\text{poly}(n)})$, which means that once the the number of actions increases, the probabilities decrease and eventually at some point every of the probability vector will have the same value which corresponds to the clipping. That is the reason for the decrease in variance with the number of actions since changing the parameters leads eventually to no further change in the output probabilities. 
With respect to the variance as a function of the number of qubits, it was predicted in Lemma \ref{lemma: lower bound contiguous-like} that for a polynomial number of actions $|A| \in \mathcal{O}(\text{poly(n)})$ the variance should decay at most polylogarithmically with the number of qubits. Moreover, if the number of actions $|A|=n$ this correponds exactly to measuring $\log(n)$ qubits and thus it was predicted that the variance should also decrease at most polynomially with the number of qubits. Such prediction is indeed confirmed through the scaling of the variance presented in the semi-logarithmic plot of Figure \ref{fig: variances_contiguous-like_clamp}(b) which illustrates the variance for a polynomially fixed number of actions.

\subsubsection*{Parity-like Born policy}
Let us now do a similar analysis for the case dealing with a parity-like Born policy. Recall that in this setting, the policy is obtained always from a global measurement. Thus, it was predicted in Section \ref{sec: pg variance} that the upper bound for the variance already would be exponentially vanishing providing the number of actions is at most polynomial $|A| \in \mathcal{O}(\text{poly(n)})$ proving that the Parity-like policy indeed suffers from the BP phenomenon. Such conclusion came from the fact that probabilities can also be considered as polynomially vanishing with the number of qubits. Outside a polynomial number of actions, the probabilities would be too small for the polynomial assumption to be valid and thus the upper bound would lose its meaning. As predicted, the variance for this policy, illustrated in Figure \ref{fig: variances_parity-like_no_clamp}(a) for non-clipped probabilities, increases with the number of actions. 
\begin{figure}[!htb]
    \centering
    \begin{minipage}[t]{\textwidth}
    \centering
    \includegraphics[height=0.20\textheight,width=\linewidth]{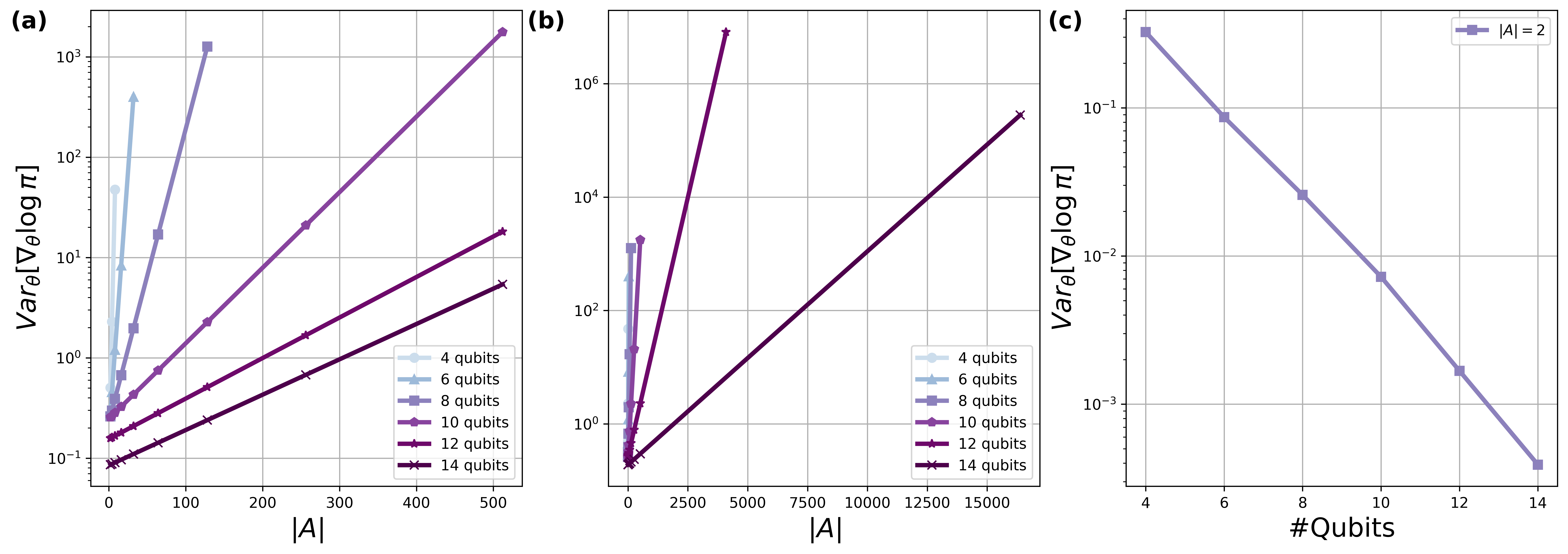}
    \end{minipage}%
    \hfill
    \caption[]{variance of the log policy gradient for parity-like Born policies: (a) and (b) as a function of $|A|$ and (c) in a semi-logarithmic plot for varying number of qubits.}
    \label{fig: variances_parity-like_no_clamp}
\end{figure}
Furthermore, Figure \ref{fig: variances_parity-like_no_clamp}(c) shows that for the base case of $|A|=2$, illustrated in semi-logarithmic plot, that the variance diminishes exponentially with the number of qubits, confirming the susceptibility of the parity-like Born policy to BPs even with a minimal number of actions. However, as depicted in Figure \ref{fig: variances_parity-like_no_clamp}(b), the variance escalates for a fixed number of actions, deviating from a vanishing trend with an increase in qubit numbers, thus hinting at exploding gradients rather than BPs. Let us now analyze the scenario in which the probabilities are clipped.
\begin{figure}[!htb]
    \centering
    \begin{minipage}[t]{\textwidth}
    \centering
    \includegraphics[height=0.20\textheight,width=\linewidth]{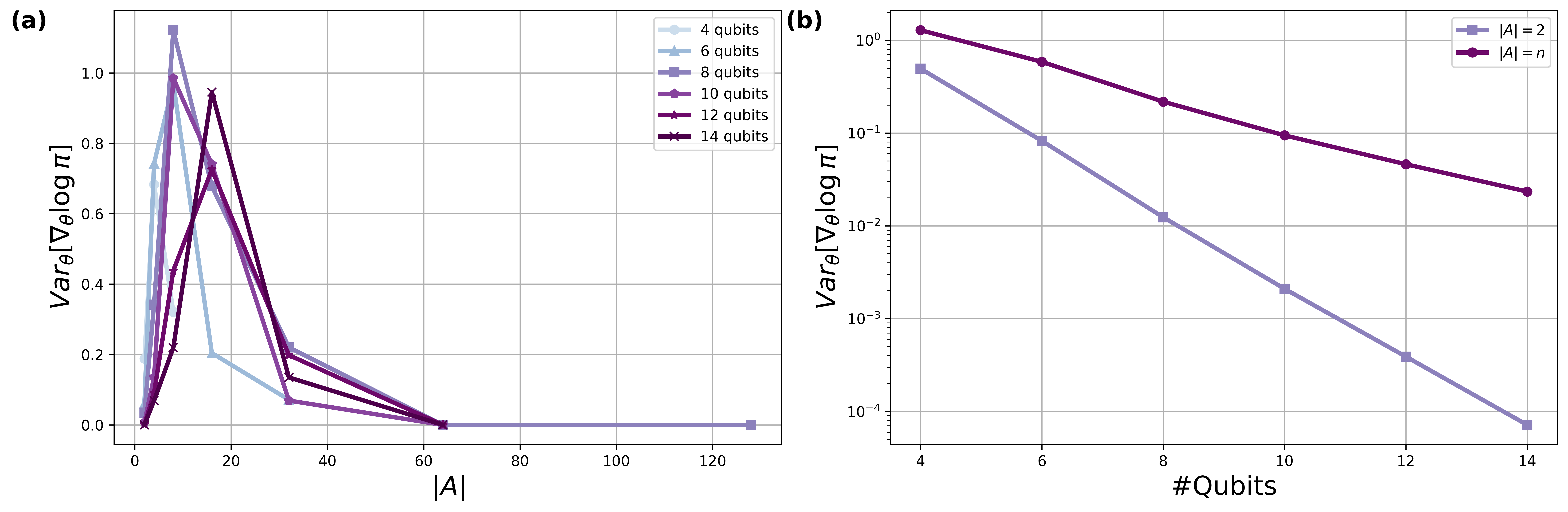}
    \end{minipage}%
    \hfill
    \caption[]{variance of the log policy gradient for parity-like Born policies as (a) a function of $|A|$ and (b) for varying number of qubits with polynomial clipping.}
    \label{fig: variances_parity-like_clamp}
\end{figure}
In Figure \ref{fig: variances_parity-like_clamp}(a), we observe a similar trend in variance as a function of the number of actions for parity-like policies with polynomially clipped probabilities. The variance initially increases with the number of actions but then undergoes a sharp decline, due to the equalization of probability values due to clipping. This behavior aligns with the hypothesis that clipping leads to uniform probabilities at higher action counts, reducing the variance. Figure \ref{fig: variances_parity-like_clamp}(b) confirms the exponential decrease in variance with the number of qubits for polynomially large actions, consistent with the theoretical predictions made in Section \ref{sec: pg variance}.
\subsubsection*{Analysis of the FIM spectrum}

The Fisher Information Matrix (FIM) spectrum analysis, as discussed in Section \ref{subsec: analysis of the Fisher Information spectrum}, is crucial for identifying trainability issues in different policy types. Focusing on the parity-like policy first, if this policy is prone to barren plateaus (BPs) with a polynomial number of actions, as inferred from previous sections, then the FIM entries should diminish exponentially with the number of qubits. This would confirm the presence of a BP, characterized by eigenvalues concentrated around zero. Conversely, beyond a polynomial number of actions and without probability clipping, the variance and subsequently the FIM entries would increase, leading to a less concentrated eigenvalue spectrum around zero, indicating the absence of a BP. These predicted behaviors are indeed evident in the FIM spectrum for the parity-like policy, as shown in Figure \ref{fig: FIM spectrum parity-like}.

\begin{figure}[!htb]
    \centering
    \begin{minipage}[t]{\textwidth}
    \centering
    \includegraphics[height=0.20\textheight,width=\linewidth]{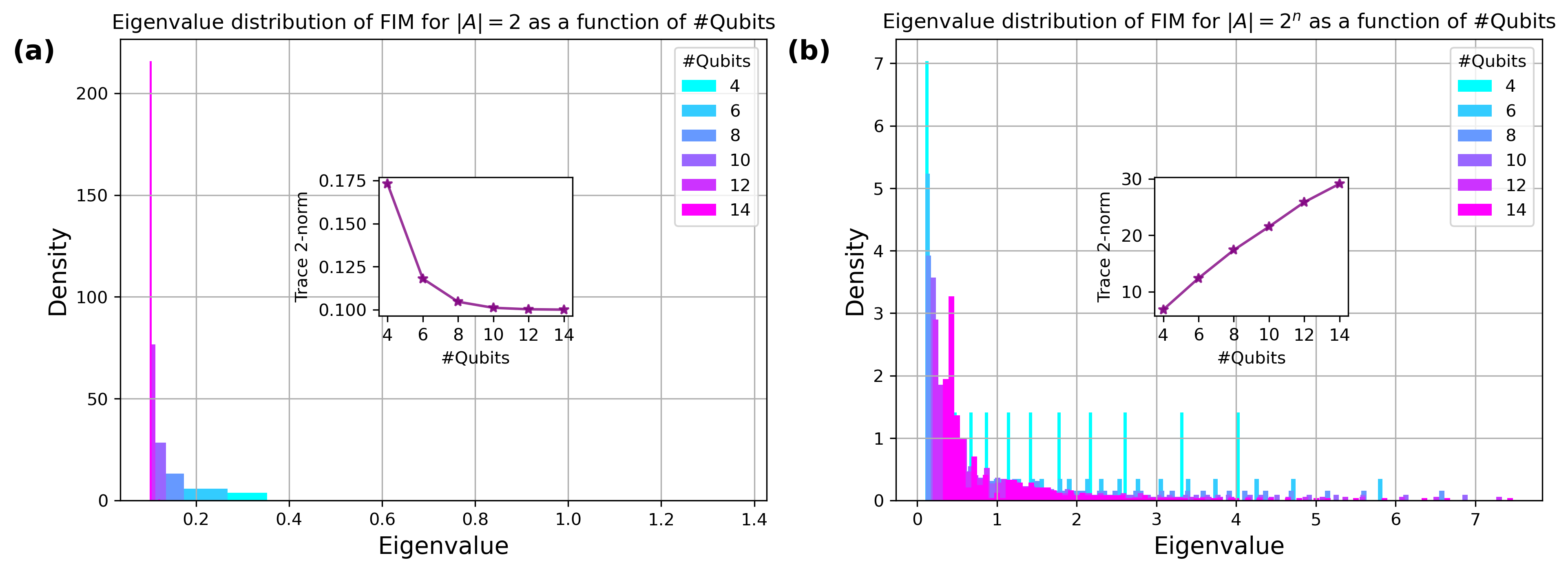}
    \end{minipage}%
    \hfill
    \caption[]{Eigenvalue distribution for the FIM as a function of the number of qubits for the parity-like Born policy in the extreme cases of $|A|=2$ and $|A|=2^n$.}
    \label{fig: FIM spectrum parity-like}
\end{figure}
In Figure \ref{fig: FIM spectrum parity-like}(a) it can be observed that for $|A|=2$ eigenvalues concentrate in zero as the number of qubits increase, confirming the presence of a BP. Moreover, in Figure \ref{fig: FIM spectrum parity-like}(b) it can be observed the opposite behavior for $|A|=2^n$. The eigenvalues are moving away from zero as the number of qubits increase, thus confirming that outside a polynomial number of actions, since the variance increases with the number of qubits the FIM spectrum does not indicate the presence of a BP.\\
Let us now inspect the spectrum of the FIM for the contiguous-like policy.
\begin{figure}[!htb]
    \centering
    \begin{minipage}[t]{\textwidth}
    \centering
    \includegraphics[height=0.20\textheight,width=\linewidth]{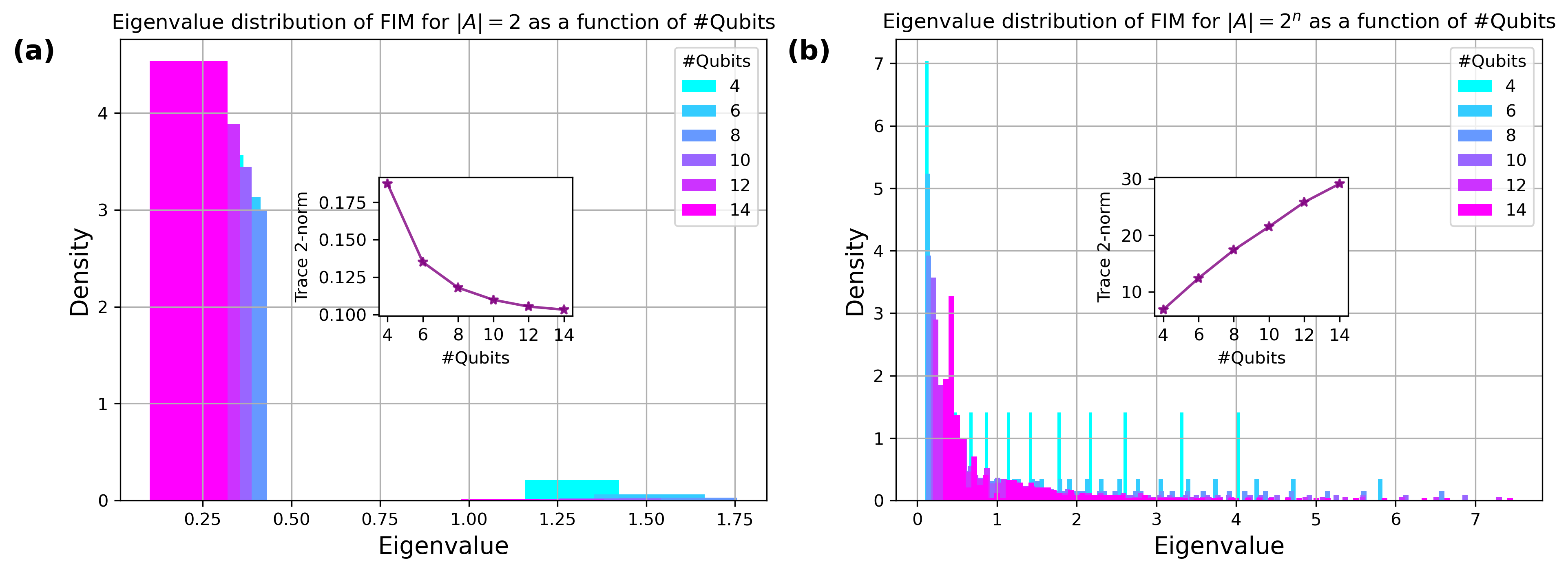}
    \end{minipage}%
    \hfill
    \caption[]{Eigenvalue distribution for the FIM as a function of the number of qubits for the contiguous-like Born policy in the extreme cases of $|A|=2$ and $|A|=2^n$.}
    \label{fig: FIM spectrum contiguous-like}
\end{figure}
The contiguous-like policy exhibits a distinct variance trend compared to the parity-like policy, especially for a smaller number of polynomial actions. Unlike the parity-like policy, this variance decreases polynomially with an increase in the number of qubits, suggesting the absence of a BP. Consequently, the FIM spectrum is expected to show eigenvalues concentrating around zero as the number of qubits increases, albeit more gradually than in the parity-like policy. However, beyond polynomial action numbers, similar to the parity-like policy, the eigenvalues will deviate from zero as the number of qubits increases due to the rising variance. Figure \ref{fig: FIM spectrum contiguous-like} corroborates these predictions. Particularly, Figure \ref{fig: FIM spectrum contiguous-like}(a) illustrates the eigenvalue concentration around zero for increasing qubits, but with a noticeably lower density compared to the parity-like case shown in Figure \ref{fig: FIM spectrum parity-like}(a).
\subsection{Multi armed bandits}
This section discusses trainability in the context of PQC-based policies, particularly what quantum systems with a substantial number of qubits, and evaluates their effectiveness in learning optimal actions in a reward-based context akin to Reinforcement Learning (RL). We consider a multi-armed bandit environment featuring a simple linear reward function: for each arm $a$, the deterministic reward $R(a)$ is given by $R(a) = 2a$. This setup allows us to scrutinize the learning capabilities of PQC-based policies with respect to the number of available actions. The PQC-based policy architecture we examine comprises a single layer of $\sigma_z$ and $\sigma_y$ single-qubit rotations, followed by an all-to-all $CZ$ entanglement pattern. For these experiments, we use 16 qubits ($n=16$) to gauge the impact of PQC depth on trainability in the scenario we have a contiguous or a parity-like policy. We analyze two scenarios: 1) a bandit environment with $|A|=n$ arms, which results in a contiguous-like policy that is $\log(n)$-local, and 2) a bandit environment with $|A|=2^{n-4}$ arms, leading to a contiguous-like policy involving measurements over a polynomial number of qubits. The performance of contiguous-like policies considered in those scenarios is compared with the global parity-like policy with the same number of qubits. Each scenario incorporates a polynomial number of measurements. Gradient estimation is conducted using parameter-shift rules.

We assess the performance of the PQC-based agent by tracking the probability of choosing the best arm over a fixed number of episodes. An episode in this bandit environment involves performing a single action, collecting the associated reward, and using this information to update the PQC-based policy parameters via gradient-based methods. We conduct 100 episodes, comprising 100 action-steps, and average the probability of selecting the best arm over 50 different trials with randomly selected parameters. 
\begin{figure}[!htb]
    \centering
    \begin{minipage}[t]{\textwidth}
    \centering
    \includegraphics[height=0.20\textheight,width=\linewidth]{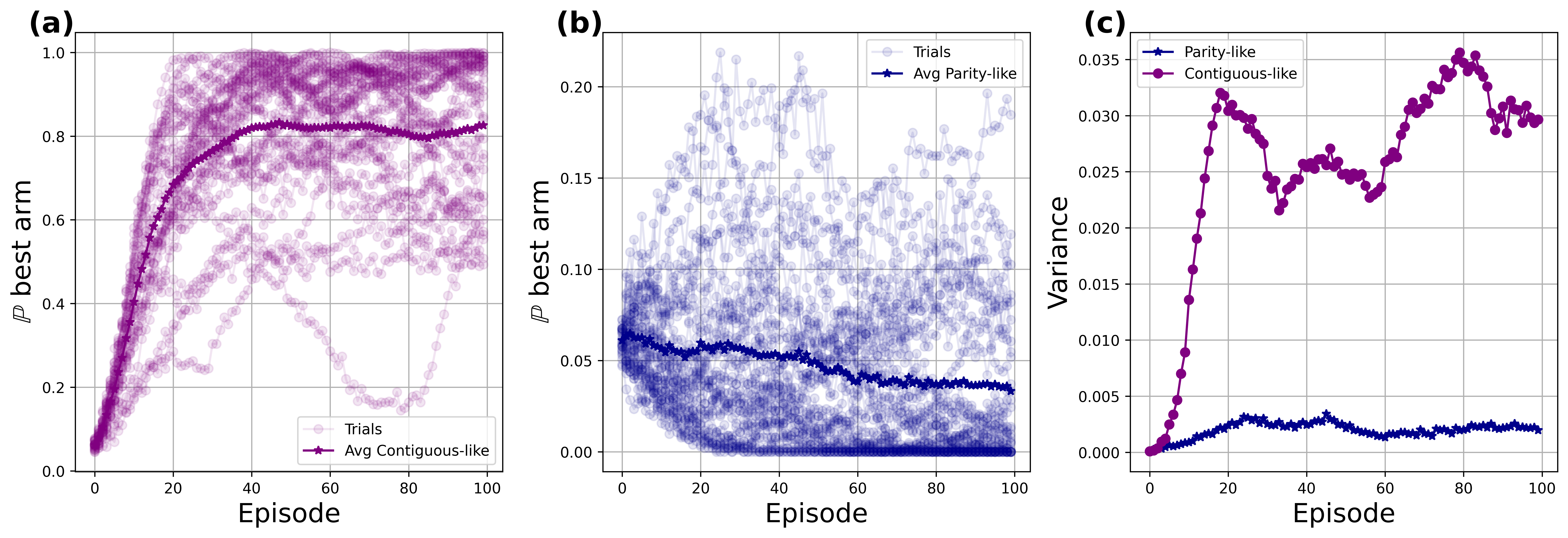}
    \end{minipage}%
    \hfill
    \caption[]{Experimental results for the Multi armed bandit environment with $|A|=n$ actions. In (a) and (b) the probability of selecting the best arm is illustrated for the contiguous-like and parity-like policies, respectively. (c) illustrates the variance of the log policy gradient for both policies.}
    \label{fig: quantum_bandits a=n}
\end{figure}
Figure \ref{fig: quantum_bandits a=n} presents the outcomes when $|A|=n$ arms. In subfigures \ref{fig: quantum_bandits a=n}(a) and \ref{fig: quantum_bandits a=n}(b), the probability of choosing the best arm is depicted for contiguous-like and parity-like policies, respectively. The contiguous-like policy generally achieves a selection probability above 0.8 for the best arm, with some instances reaching a deterministic policy (probability of 1). The parity-like policy, however, struggles to exceed a 0.5 selection probability throughout training. This disparity can be attributed to the differing localities of each policy. Since $|A|=n$, the contiguous-like policy involves $\log(n)$-local measurements, contrasting with the parity-like policy's $n$-local approach. The effect of these differing observables on trainability is further illuminated by examining the variance of the log policy gradient, as shown in Figure \ref{fig: quantum_bandits a=n}(c). While the variance remains relatively low for both policies, it is notably smaller and close to zero for the parity-like policy.
\begin{figure}[!htb]
    \centering
    \begin{minipage}[t]{\textwidth}
    \centering
    \includegraphics[height=0.20\textheight,width=\linewidth]{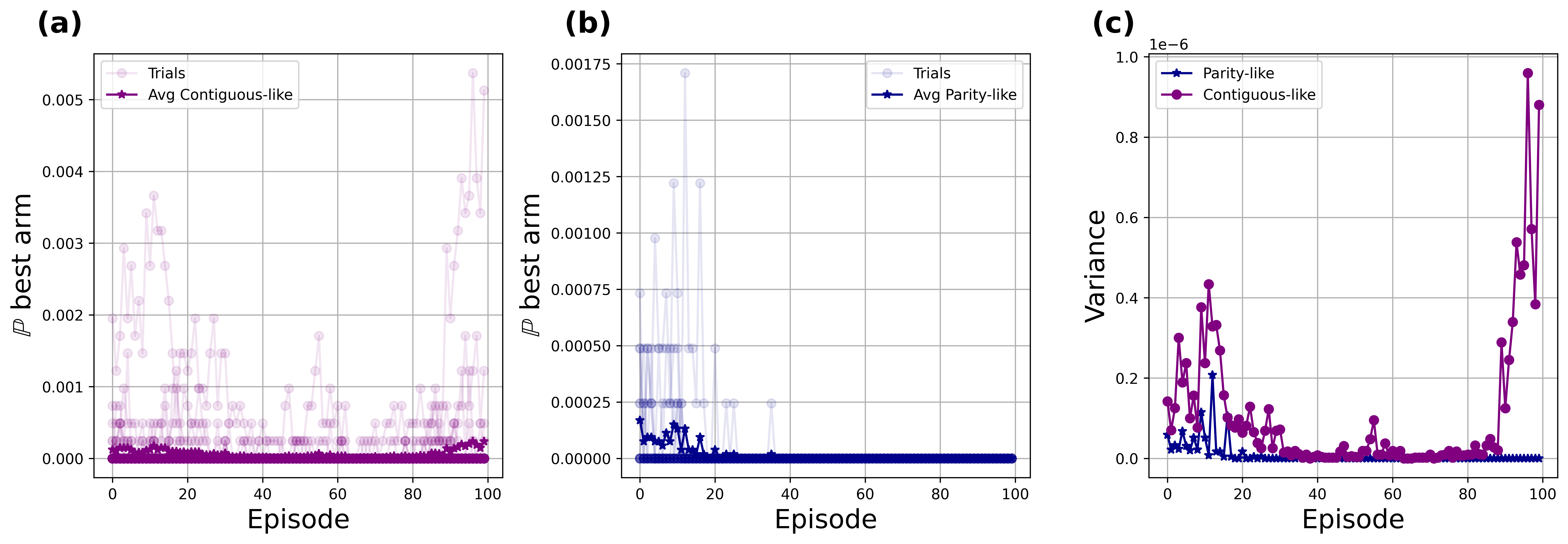}
    \end{minipage}%
    \hfill
    \caption[]{Experimental results for the Multi armed bandit environment with $|A|=2^{n-4}$ actions. In (a) and (b) the probability of selecting the best arm is illustrated for the contiguous-like and parity-like policies, respectively. (c) illustrates the variance of the log policy gradient for both policies.}
    \label{fig: quantum_bandits a=2n-4}
\end{figure}
Figure \ref{fig: quantum_bandits a=2n-4} depicts results for the bandit environment with $|A|=2^{n-4}$ arms. Subfigures \ref{fig: quantum_bandits a=2n-4}(a) and \ref{fig: quantum_bandits a=2n-4}(b) depict the probability of selecting the best arm over a series of episodes for both contiguous-like and parity-like policies. In this scenario, neither policy demonstrates the ability to learn the optimal arm effectively, with probabilities of selecting the best arm consistently below one percent. This outcome is explained by examining the employed observables. The parity-like policy is always globally measured, but now, with more arms, the probabilities associated with each arm are significantly reduced. Furthermore, the contiguous-like policy, which previously measured $\log(n)$ qubits, engages in a measurement over a polylog number of qubits. Despite not experiencing exponential decay in variance, the large number of qubits and actions places the contiguous-like policy in a barren plateau, akin to the parity-like policy. This is further evidenced by the variance of the gradient, as shown in Figure \ref{fig: quantum_bandits a=2n-4}(c). The variance for both policies is minimal, indicating an inability to learn the optimal policy. Thus, a polynomial number of measurements, as utilized in these experiments, proves insufficient for learning the optimal policy in such complex settings.

\section{Conclusion}\label{sec: conclusions}

In conclusion, our research provides pivotal insights into the trainability of PQC-based policies in the realm of policy-based RL. A significant aspect of our findings concerns the trainability challenges faced by two specific types of policies: the Contiguous-like and Parity-like Born policies. These challenges manifest in two distinct forms: the occurrence of standard BPs characterized by exponentially diminishing gradients, and the potential for gradient explosion.

The nature and extent of these challenges are influenced by two key factors: the specific type of Born policy used, and the interplay between the number of qubits and the action-space size. Notably, our study reveals that with $n$ qubits, when a polynomial number of measurements ($\mathcal{O}(\text{poly}(n))$) is considered, the Contiguous-like Born policy demonstrates trainable regions at a logarithmic depth ($\mathcal{O}(\text{log}(n))$), provided that the action-space remains within polynomial bounds ($\mathcal{O}(\text{poly}(n))$). This observation is crucial as it delineates a specific scenario where this policy remains effective and trainable. In contrast, under similar conditions, the Parity-like Born policy consistently exhibits a BP, indicating its inherent limitations in certain settings. Furthermore, we noticed a striking shift in gradient behavior for actions beyond polynomial size. In these instances, the gradients transition from diminishing to exploding, attributed to exceedingly low probabilities. This shift renders the polynomial number of measurements inadequate for differentiating actions, thus presenting a significant challenge to the practical application of these policies.

We would like to emphasize that part of the gradient behavior observed throughout this work is due to classical post-processing. Recall that in Section \ref{subsec: product states} we analyzed the variance for the trivial scenario of product states. We concluded that product states would not lead to trainability issues in policy gradient optimization. This is given by the score function, or the logarithm post-processing of the probability of basis states, that separates the product of individual qubit contributions into a sum. Moreover, we expect the classical policy gradient algorithm to suffer from the same problem. Indeed, for very small probabilities, the gradient would also explode. This is one reason advanced policy gradient algorithms such as PPO \cite{schulmanProximalPolicyOptimization2017a} are much more stable to train. The crucial difference compared with PQC-based policies stems from the fact that the probabilities derived from quantum systems will eventually concentrate given more expressivity and depth of the circuit \cite{rudolphTrainabilityBarriersOpportunities2023}, leading to other sorts of optimization problems that we do not know to be as severe in the classical setting. In addition, we would also like to stress that the trainability analysis presented in this work neglected the effect of the reward. Indeed, it was assumed the presence of a maximum reward to simplify the policy gradient REINFORCE objective expressed in Equation \eqref{eq: policy gradient baseline} to an expression dependent only on the policy. Nevertheless, in practice, there are problems such as the sparsity of the reward. For instance, environments where the reward is assigned only at a goal state and no reward in the middle. In such scenarios, small or even null rewards would lead to vanishing gradients and, indeed, hide the behavior of the quantum system. For that matter, we did not consider these cases.

In \cite{cerezoDoesProvableAbsence2024}, the authors showed that the BP phenomenon results from a curse of dimensionality and that the cases where one can impose trainability guarantees also lead to classically simulable models. We suspect that this is most likely the case for PQC-based policies since we are still considering hardware-efficient ansatze, and under the measurement conditions outlined through this work, it would fall under the same category explored in \cite{cerezoDoesProvableAbsence2024}. However, it raises the question of whether there are other types of ansatze that strike a perfect balance between trainability and non-efficient classical simulation, combined with specific PQC-based optimization that could be used to circumvent this issue.

There are several other promising directions for future work. For instance, one should address the trainability issues associated with softmax policies \cite{jerbi_variational_2021,sequeira_policy_2023}, where the freedom to measure the expectation values of $|A|$ different observables presents an intriguing avenue for exploration since more sophisticated results in the trainability of PQCs \cite{ragoneUnifiedTheoryBarren2023} can be considered.

\section*{Acknowledgements}\label{sec: acknowledgements}

This work is financed by National Funds through the Portuguese funding agency, FCT — Fundação para a Ciência e a Tecnologia, within grants LA/P/0063/2020, UI/BD/152698/2022 and project IBEX, with reference PTDC/CC1-COM/4280/2021

\appendix
\section{Upper bound on the return}\label{appendix: upper bound return}

Let $R_{max}$ be the maximum possible reward at any time step. Thus, the return is upper bounded by:

\begin{equation}
G(\tau) = \sum_{t=0}^{T-1} \gamma^t r_{t+1} \leq R_{max} \sum_{t=0}^{T-1} \gamma^t = R_{max} \frac{\gamma^T-1}{\gamma-1} 
\end{equation}
\noindent
This enables the following upper bound on the return per time step.

\begin{equation}
 \sum_{t=0}^{T-1} G_t(\tau) \leq R_{max} \sum_{t=0}^{T-1} \frac{\gamma^{T-t} - 1}{(\gamma - 1)} \leq R_{max} \frac{T}{(\gamma - 1)^2}
\label{eq: upper_bound_return}
\end{equation}

\section{Proof of Lemma \ref{lemma: lower bound contiguous-like}}\label{appendix: lower bound contiguous-like}
    \bornpolicygradientlemma*

    \begin{proof}
    
        Let us start with the expansion of the standard expression of the variance. For the sake of simplicity let $\pi_{\theta} = \pi(a|s,\theta)$ and the partial derivative $\partial_{\theta} \log \pi_{\theta} = \frac{\partial_{\theta} \pi_{\theta}}{\pi_{\theta}}$.
        \begin{alignLetter}
            \mathbb{V}_{\theta}\biggl[\partial_{\theta} \log \pi_{\theta} \biggr] &= \mathbb{E}_{\theta}\biggl[\biggl(\frac{\partial_{\theta} \pi_{\theta}}{\pi_{\theta}}\biggr)^2\biggr] - \mathbb{E}_{\theta}\biggl[\frac{\partial_{\theta} \pi_{\theta}}{\pi_{\theta}}\biggr]^2 \nonumber \\
            &\geq \mathbb{E}_{\theta}\biggl[(\partial_{\theta} \pi_{\theta})^2\biggr] \mathbb{E}_{\theta}\biggl[\frac{1}{\pi_{\theta}^2}\biggr] - \mathbb{V}_{\theta}\biggl[(\partial_{\theta} \pi_{\theta})^2\biggr] \mathbb{V}_{\theta}\biggl[\frac{1}{\pi_{\theta}^2}\biggr] - \mathbb{E}_{\theta}\biggl[\frac{\partial_{\theta} \pi_{\theta}}{\pi_{\theta}}\biggr]^2 \\
            &\geq \mathbb{E}_{\theta}\biggl[(\partial_{\theta} \pi_{\theta})^2\biggr] \mathbb{E}_{\theta}\biggl[\frac{1}{\pi_{\theta}^2}\biggr] - \mathbb{V}_{\theta}\biggl[(\partial_{\theta} \pi_{\theta})^2\biggr] \mathbb{V}_{\theta}\biggl[\frac{1}{\pi_{\theta}^2}\biggr] - \mathbb{V}_{\theta}\biggl[\partial_{\theta} \pi_{\theta}\biggr] \mathbb{V}_{\theta}\biggl[\frac{1}{\pi_{\theta}}\biggr]\\
            &= \mathbb{E}_{\theta}\biggl[(\partial_{\theta} \pi_{\theta})^2\biggr] \mathbb{E}_{\theta}\biggl[\frac{1}{\pi_{\theta}^2}\biggr] - \underbrace{\biggl( \mathbb{V}_{\theta}\biggl[(\partial_{\theta} \pi_{\theta})^2\biggr] \mathbb{V}_{\theta}\biggl[\frac{1}{\pi_{\theta}^2}\biggr] + \mathbb{V}_{\theta}\biggl[\partial_{\theta} \pi_{\theta}\biggr] \mathbb{V}_{\theta}\biggl[\frac{1}{\pi_{\theta}}\biggr] \biggr)}_{(a)} \nonumber
        \end{alignLetter}
        where (A) is obtained from the lower bound of the expectation value of the product of two non-negative random variables $\mathbb{E}_{\theta}[XY ] \geq \mathbb{E}_{\theta}[X]\mathbb{E}_{\theta}[Y] - \mathbb{V}_{\theta}[X]\mathbb{V}_{\theta}[Y]$ and (B) from the upper bound of the variance of the product of two random variables via Cauchy-Schwarz $\mathbb{V}_{\theta}[XY] \leq \sqrt{\mathbb{V}_{\theta}[X]\mathbb{V}_{\theta}[Y]}$ \cite{thanasilp_subtleties_2021}. The variance is lower bounded taking the upper bound of $(a)$ that can be simplified to:
        \begin{alignLetter}
            (a) &\leq \biggl( 2\mathbb{V}_{\theta}\biggl[ \partial_{\theta} \pi_{\theta}\biggr] \biggl| \partial_{\theta} \pi_{\theta} \biggr|_{\text{max}}^2 + 2\mathbb{E}_{\theta} \biggl[ \partial_{\theta} \pi_{\theta}\biggr]\mathbb{V}_{\theta}\biggl[ \partial_{\theta} \pi_{\theta}\biggr] \biggr)  \biggl| \frac{1}{\pi_{\theta}^2} \biggr|_{\text{max}} + \mathbb{V}_{\theta}\biggl[\partial_{\theta} \pi_{\theta}\biggr] \mathbb{V}_{\theta}\biggl[\frac{1}{\pi_{\theta}}\biggr] \\
            &\leq \frac{1}{2}\mathbb{V}_{\theta}\biggl[ \partial_{\theta} \pi_{\theta}\biggr] \biggl| \frac{1}{\pi_{\theta}^2} \biggr|_{\text{max}} + \mathbb{V}_{\theta}\biggl[\partial_{\theta} \pi_{\theta}\biggr] \mathbb{V}_{\theta}\biggl[\frac{1}{\pi_{\theta}}\biggr] \\
            &\leq \frac{3}{2}\mathbb{V}_{\theta}\biggl[ \partial_{\theta} \pi_{\theta}\biggr] \biggl| \frac{1}{\pi_{\theta}^2} \biggr|_{\text{max}} \\
        \end{alignLetter}
        where (A) is obtained from the upper bound of the variance pf the product of two random variables, (B) from the assumption that either parameterized block before/after $\theta$ forms a 1-design and thus $\mathbb{E}_{\theta} [ \partial_{\theta} \pi_{\theta}]=0$ and (C) from the upper bound on the variance $\mathbb{V}_{\theta}[\frac{1}{\pi_{\theta}}] \leq | \frac{1}{\pi_{\theta}^2}|_{\text{max}}$.

        The lower bound on the variance of the policy gradient can thus be further simplified to:
        \begin{alignLetter}
            \mathbb{V}_{\theta}\biggl[\partial_{\theta} \log \pi_{\theta} \biggr] &\geq \mathbb{E}_{\theta}\biggl[(\partial_{\theta} \pi_{\theta})^2\biggr] \mathbb{E}_{\theta}\biggl[\frac{1}{\pi_{\theta}^2}\biggr] - \frac{3}{2}\mathbb{V}_{\theta}\biggl[ \partial_{\theta} \pi_{\theta}\biggr] \biggl| \frac{1}{\pi_{\theta}^2} \biggr|_{\text{max}}\nonumber \\
            &\geq \biggl(\mathbb{E}_{\theta}\biggl[\partial_{\theta} \pi_{\theta}\biggr]^2 - \mathbb{V}_{\theta}\biggl[\partial_{\theta} \pi_{\theta}\biggr]^2 \biggr)\biggl(\mathbb{E}_{\theta}\biggl[\frac{1}{\pi_{\theta}}\biggr]^2 - \mathbb{V}_{\theta}\biggl[\frac{1}{\pi_{\theta}}\biggr]^2 \biggr) - \frac{3}{2}\mathbb{V}_{\theta}\biggl[ \partial_{\theta} \pi_{\theta}\biggr] \biggl| \frac{1}{\pi_{\theta}} \biggr|_{\text{max}}^2 \\
            &=  \mathbb{V}_{\theta}\biggl[\partial_{\theta} \pi_{\theta}\biggr]^2 \mathbb{V}_{\theta}\biggl[\frac{1}{\pi_{\theta}}\biggr]^2 - \biggl(\mathbb{V}_{\theta}\biggl[\partial_{\theta} \pi_{\theta}\biggr]^2 \mathbb{E}_{\theta}\biggl[\frac{1}{\pi_{\theta}}\biggr]^2 + \frac{3}{2}\mathbb{V}_{\theta}\biggl[ \partial_{\theta} \pi_{\theta}\biggr] \biggl| \frac{1}{\pi_{\theta}} \biggr|_{\text{max}}^2\biggr) \\
            &\geq \mathbb{V}_{\theta}\biggl[\partial_{\theta} \pi_{\theta}\biggr]^2 \mathbb{V}_{\theta}\biggl[\frac{1}{\pi_{\theta}}\biggr]^2 - 3\mathbb{V}_{\theta}\biggl[\partial_{\theta} \pi_{\theta}\biggr]^2 \mathbb{E}_{\theta}\biggl[\frac{1}{\pi_{\theta}}\biggr]^2 \\
            &= \mathbb{V}_{\theta}\biggl[\partial_{\theta} \pi_{\theta}\biggr]^2 \underbrace{\biggl(\mathbb{V}_{\theta}\biggl[\frac{1}{\pi_{\theta}}\biggr]^2 - 3\mathbb{E}_{\theta}\biggl[\frac{1}{\pi_{\theta}}\biggr]^2 \biggr)}_{(a)} \nonumber \\
        \end{alignLetter}
        where (A) is obtained from the lower bound of the expectation value of the product of two non-negative random variables, (B) from the assumption that either parameterized block before/after $\theta$ forms a 1-design and thus $\mathbb{E}_{\theta} [ \partial_{\theta} \pi_{\theta}]=0$ and reorganizing terms and (C) from the upper bound on the expectation value and joining terms.

        Since the variance is non-negative it implies that $(a) \geq 0$. Therefore the variance will be lower bounded depending on the number of actions and corresponding globality of the observable. For $|A| \in \mathcal{O}(n)$, $\mathbb{V}_{\theta}\biggl[\partial_{\theta} \pi_{\theta}\biggr]^2 \in \Omega(\frac{1}{\text{poly}(n)}^2)$ for $\mathcal{O}(\log(n))$ depth. It decays polynomially with the number of qubits since we are measuring $\log(n)$ (adjacent) qubits \cite{cerezo_cost_2021}. Moreover, $(a) \leq \text{poly}(n)^2$. Thus, the overall variance deacay at most polynomially with the number of qubits. When the number of actions $|A| \in \mathcal{O}(\text{poly}(n))$, $\mathbb{V}_{\theta}\biggl[\partial_{\theta} \pi_{\theta}\biggr]^2 \in \Omega(2^{-\text{poly}(\log(n))^2})$. It decays faster than polynomially but slower than exponentially since we are measuring $\log(\text{poly}(n))$ qubits \cite{cerezo_cost_2021}. In this case $(a) \leq 2^{\text{poly}(\log(n))^2}$ since we have $|A| \in \text{poly}(n)$. Therefore the overall variance decay at most polylogarithmically with the number of qubits. Thus, completing the proof. 
    \end{proof}

\bibliographystyle{plainnat}
\bibliography{refs,BP_QPG}

\begin{thebibliography}{23}
\providecommand{\natexlab}[1]{#1}
\providecommand{\url}[1]{\texttt{#1}}
\expandafter\ifx\csname urlstyle\endcsname\relax
  \providecommand{\doi}[1]{doi: #1}\else
  \providecommand{\doi}{doi: \begingroup \urlstyle{rm}\Url}\fi

\bibitem[Abbas et~al.(2021)Abbas, Sutter, Zoufal, Lucchi, Figalli, and Woerner]{abbasPowerQuantumNeural2021}
Amira Abbas, David Sutter, Christa Zoufal, Aur{\'e}lien Lucchi, Alessio Figalli, and Stefan Woerner.
\newblock The power of quantum neural networks.
\newblock \emph{Nature Computational Science}, 1\penalty0 (6):\penalty0 403--409, June 2021.
\newblock ISSN 2662-8457.
\newblock \doi{10.1038/s43588-021-00084-1}.

\bibitem[Arrasmith et~al.(2021)Arrasmith, Cerezo, Czarnik, Cincio, and Coles]{arrasmith_effect_2021}
Andrew Arrasmith, M.~Cerezo, Piotr Czarnik, Lukasz Cincio, and Patrick~J. Coles.
\newblock Effect of barren plateaus on gradient-free optimization.
\newblock \emph{Quantum}, 5:\penalty0 558, October 2021.
\newblock ISSN 2521-327X.
\newblock \doi{10.22331/q-2021-10-05-558}.
\newblock URL \url{http://arxiv.org/abs/2011.12245}.
\newblock arXiv:2011.12245 [quant-ph, stat].

\bibitem[Bergholm et~al.(2022)Bergholm, Izaac, Schuld, Gogolin, Ahmed, Ajith, Alam, {Alonso-Linaje}, AkashNarayanan, Asadi, Arrazola, Azad, Banning, Blank, Bromley, Cordier, Ceroni, Delgado, Di~Matteo, Dusko, Garg, Guala, Hayes, Hill, Ijaz, Isacsson, Ittah, Jahangiri, Jain, Jiang, Khandelwal, Kottmann, Lang, Lee, Loke, Lowe, McKiernan, Meyer, {Monta{\~n}ez-Barrera}, Moyard, Niu, O'Riordan, Oud, Panigrahi, Park, Polatajko, Quesada, Roberts, S{\'a}, Schoch, Shi, Shu, Sim, Singh, Strandberg, Soni, Sz{\'a}va, Thabet, {Vargas-Hern{\'a}ndez}, Vincent, Vitucci, Weber, Wierichs, Wiersema, Willmann, Wong, Zhang, and Killoran]{bergholmPennyLaneAutomaticDifferentiation2022}
Ville Bergholm, Josh Izaac, Maria Schuld, Christian Gogolin, Shahnawaz Ahmed, Vishnu Ajith, M.~Sohaib Alam, Guillermo {Alonso-Linaje}, B.~AkashNarayanan, Ali Asadi, Juan~Miguel Arrazola, Utkarsh Azad, Sam Banning, Carsten Blank, Thomas~R. Bromley, Benjamin~A. Cordier, Jack Ceroni, Alain Delgado, Olivia Di~Matteo, Amintor Dusko, Tanya Garg, Diego Guala, Anthony Hayes, Ryan Hill, Aroosa Ijaz, Theodor Isacsson, David Ittah, Soran Jahangiri, Prateek Jain, Edward Jiang, Ankit Khandelwal, Korbinian Kottmann, Robert~A. Lang, Christina Lee, Thomas Loke, Angus Lowe, Keri McKiernan, Johannes~Jakob Meyer, J.~A. {Monta{\~n}ez-Barrera}, Romain Moyard, Zeyue Niu, Lee~James O'Riordan, Steven Oud, Ashish Panigrahi, Chae-Yeun Park, Daniel Polatajko, Nicol{\'a}s Quesada, Chase Roberts, Nahum S{\'a}, Isidor Schoch, Borun Shi, Shuli Shu, Sukin Sim, Arshpreet Singh, Ingrid Strandberg, Jay Soni, Antal Sz{\'a}va, Slimane Thabet, Rodrigo~A. {Vargas-Hern{\'a}ndez}, Trevor Vincent, Nicola Vitucci, Maurice Weber, David Wierichs, Roeland Wiersema, Moritz Willmann, Vincent Wong, Shaoming Zhang, and Nathan Killoran.
\newblock {{PennyLane}}: {{Automatic}} differentiation of hybrid quantum-classical computations, July 2022.

\bibitem[Biamonte(2021)]{biamonte_universal_2021}
Jacob Biamonte.
\newblock Universal variational quantum computation.
\newblock \emph{Physical Review A}, 103\penalty0 (3):\penalty0 L030401, March 2021.
\newblock ISSN 2469-9926, 2469-9934.
\newblock \doi{10.1103/PhysRevA.103.L030401}.
\newblock URL \url{https://link.aps.org/doi/10.1103/PhysRevA.103.L030401}.

\bibitem[Cerezo et~al.(2021{\natexlab{a}})Cerezo, Arrasmith, Babbush, Benjamin, Endo, Fujii, McClean, Mitarai, Yuan, Cincio, and Coles]{cerezo_variational_2021}
M.~Cerezo, Andrew Arrasmith, Ryan Babbush, Simon~C. Benjamin, Suguru Endo, Keisuke Fujii, Jarrod~R. McClean, Kosuke Mitarai, Xiao Yuan, Lukasz Cincio, and Patrick~J. Coles.
\newblock Variational quantum algorithms.
\newblock \emph{Nature Reviews Physics}, 3\penalty0 (9):\penalty0 625--644, September 2021{\natexlab{a}}.
\newblock ISSN 2522-5820.
\newblock \doi{10.1038/s42254-021-00348-9}.
\newblock URL \url{https://www.nature.com/articles/s42254-021-00348-9}.
\newblock Number: 9 Publisher: Nature Publishing Group.

\bibitem[Cerezo et~al.(2021{\natexlab{b}})Cerezo, Sone, Volkoff, Cincio, and Coles]{cerezo_cost_2021}
M.~Cerezo, Akira Sone, Tyler Volkoff, Lukasz Cincio, and Patrick~J. Coles.
\newblock Cost {Function} {Dependent} {Barren} {Plateaus} in {Shallow} {Parametrized} {Quantum} {Circuits}.
\newblock \emph{Nature Communications}, 12\penalty0 (1):\penalty0 1791, March 2021{\natexlab{b}}.
\newblock ISSN 2041-1723.
\newblock \doi{10.1038/s41467-021-21728-w}.
\newblock URL \url{http://arxiv.org/abs/2001.00550}.
\newblock arXiv:2001.00550 [quant-ph].

\bibitem[Cerezo et~al.(2024)Cerezo, Larocca, {Garc{\'i}a-Mart{\'i}n}, Diaz, Braccia, Fontana, Rudolph, Bermejo, Ijaz, Thanasilp, Anschuetz, and Holmes]{cerezoDoesProvableAbsence2024}
M.~Cerezo, Martin Larocca, Diego {Garc{\'i}a-Mart{\'i}n}, N.~L. Diaz, Paolo Braccia, Enrico Fontana, Manuel~S. Rudolph, Pablo Bermejo, Aroosa Ijaz, Supanut Thanasilp, Eric~R. Anschuetz, and Zo{\"e} Holmes.
\newblock Does provable absence of barren plateaus imply classical simulability? {{Or}}, why we need to rethink variational quantum computing, March 2024.

\bibitem[Chen et~al.(2020)Chen, Yang, Qi, Chen, Ma, and Goan]{chen_variational_2020}
Samuel Yen-Chi Chen, Chao-Han~Huck Yang, Jun Qi, Pin-Yu Chen, Xiaoli Ma, and Hsi-Sheng Goan.
\newblock Variational quantum circuits for deep reinforcement learning.
\newblock \emph{IEEE Access}, 8:\penalty0 141007--141024, 2020.
\newblock Publisher: IEEE.

\bibitem[Cherrat et~al.(2023)Cherrat, Raj, Kerenidis, Shekhar, Wood, Dee, Chakrabarti, Chen, Herman, Hu, Minssen, Shaydulin, Sun, Yalovetzky, and Pistoia]{cherrat_quantum_2023}
El~Amine Cherrat, Snehal Raj, Iordanis Kerenidis, Abhishek Shekhar, Ben Wood, Jon Dee, Shouvanik Chakrabarti, Richard Chen, Dylan Herman, Shaohan Hu, Pierre Minssen, Ruslan Shaydulin, Yue Sun, Romina Yalovetzky, and Marco Pistoia.
\newblock Quantum {Deep} {Hedging}, March 2023.
\newblock URL \url{http://arxiv.org/abs/2303.16585}.
\newblock arXiv:2303.16585 [quant-ph, q-fin].

\bibitem[Jerbi et~al.(2021)Jerbi, Gyurik, Marshall, Briegel, and Dunjko]{jerbi_variational_2021}
Sofiene Jerbi, Casper Gyurik, Simon Marshall, Hans~J. Briegel, and Vedran Dunjko.
\newblock Variational quantum policies for reinforcement learning.
\newblock \emph{arXiv preprint arXiv:2103.05577}, 2021.

\bibitem[Jerbi et~al.(2022)Jerbi, Cornelissen, Ozols, and Dunjko]{jerbi_quantum_2022}
Sofiene Jerbi, Arjan Cornelissen, Māris Ozols, and Vedran Dunjko.
\newblock Quantum policy gradient algorithms, December 2022.
\newblock URL \url{http://arxiv.org/abs/2212.09328}.
\newblock arXiv:2212.09328 [quant-ph, stat].

\bibitem[Leone et~al.(2022)Leone, Oliviero, Cincio, and Cerezo]{leone_practical_2022}
Lorenzo Leone, Salvatore F.~E. Oliviero, Lukasz Cincio, and M.~Cerezo.
\newblock On the practical usefulness of the {Hardware} {Efficient} {Ansatz}, November 2022.
\newblock URL \url{http://arxiv.org/abs/2211.01477}.
\newblock arXiv:2211.01477 [quant-ph].

\bibitem[McClean et~al.(2018)McClean, Boixo, Smelyanskiy, Babbush, and Neven]{mcclean_barren_2018}
Jarrod~R. McClean, Sergio Boixo, Vadim~N. Smelyanskiy, Ryan Babbush, and Hartmut Neven.
\newblock Barren plateaus in quantum neural network training landscapes.
\newblock \emph{Nature Communications}, 9\penalty0 (1):\penalty0 4812, November 2018.
\newblock ISSN 2041-1723.
\newblock \doi{10.1038/s41467-018-07090-4}.
\newblock URL \url{https://www.nature.com/articles/s41467-018-07090-4}.
\newblock Number: 1 Publisher: Nature Publishing Group.

\bibitem[Meyer et~al.(2023)Meyer, Scherer, Plinge, Mutschler, and Hartmann]{meyer_quantum_2023}
Nico Meyer, Daniel~D. Scherer, Axel Plinge, Christopher Mutschler, and Michael~J. Hartmann.
\newblock Quantum {Policy} {Gradient} {Algorithm} with {Optimized} {Action} {Decoding}, May 2023.
\newblock URL \url{http://arxiv.org/abs/2212.06663}.
\newblock arXiv:2212.06663 [quant-ph].

\bibitem[Ragone et~al.(2023)Ragone, Bakalov, Sauvage, Kemper, Marrero, Larocca, and Cerezo]{ragoneUnifiedTheoryBarren2023}
Michael Ragone, Bojko~N. Bakalov, Fr{\'e}d{\'e}ric Sauvage, Alexander~F. Kemper, Carlos~Ortiz Marrero, Martin Larocca, and M.~Cerezo.
\newblock A {{Unified Theory}} of {{Barren Plateaus}} for {{Deep Parametrized Quantum Circuits}}, September 2023.

\bibitem[Rudolph et~al.(2023)Rudolph, Lerch, Thanasilp, Kiss, Vallecorsa, Grossi, and Holmes]{rudolphTrainabilityBarriersOpportunities2023}
Manuel~S. Rudolph, Sacha Lerch, Supanut Thanasilp, Oriel Kiss, Sofia Vallecorsa, Michele Grossi, and Zo{\"e} Holmes.
\newblock Trainability barriers and opportunities in quantum generative modeling, May 2023.

\bibitem[Schuld et~al.(2019)Schuld, Bergholm, Gogolin, Izaac, and Killoran]{schuldEvaluatingAnalyticGradients2019}
Maria Schuld, Ville Bergholm, Christian Gogolin, Josh Izaac, and Nathan Killoran.
\newblock Evaluating analytic gradients on quantum hardware.
\newblock \emph{Physical Review A}, 99\penalty0 (3):\penalty0 032331, March 2019.
\newblock ISSN 2469-9926, 2469-9934.
\newblock \doi{10.1103/PhysRevA.99.032331}.

\bibitem[Schulman et~al.(2017)Schulman, Wolski, Dhariwal, Radford, and Klimov]{schulmanProximalPolicyOptimization2017a}
John Schulman, Filip Wolski, Prafulla Dhariwal, Alec Radford, and Oleg Klimov.
\newblock Proximal {{Policy Optimization Algorithms}}, August 2017.

\bibitem[Sequeira et~al.(2023)Sequeira, Santos, and Barbosa]{sequeira_policy_2023}
André Sequeira, Luis~Paulo Santos, and Luis~Soares Barbosa.
\newblock Policy gradients using variational quantum circuits.
\newblock \emph{Quantum Machine Intelligence}, 5\penalty0 (1):\penalty0 18, April 2023.
\newblock ISSN 2524-4914.
\newblock \doi{10.1007/s42484-023-00101-8}.
\newblock URL \url{https://doi.org/10.1007/s42484-023-00101-8}.

\bibitem[Skolik et~al.(2022)Skolik, Jerbi, and Dunjko]{skolik_quantum_2022}
Andrea Skolik, Sofiene Jerbi, and Vedran Dunjko.
\newblock Quantum agents in the gym: a variational quantum algorithm for deep q-learning.
\newblock \emph{Quantum}, 6:\penalty0 720, 2022.
\newblock Publisher: Verein zur Förderung des Open Access Publizierens in den Quantenwissenschaften.

\bibitem[Thanasilp et~al.(2021)Thanasilp, Wang, Nghiem, Coles, and Cerezo]{thanasilp_subtleties_2021}
Supanut Thanasilp, Samson Wang, Nhat~A. Nghiem, Patrick~J. Coles, and M.~Cerezo.
\newblock Subtleties in the trainability of quantum machine learning models, October 2021.
\newblock URL \url{http://arxiv.org/abs/2110.14753}.
\newblock arXiv:2110.14753 [quant-ph, stat].

\bibitem[Uvarov and Biamonte(2021)]{uvarovBarrenPlateausCost2021}
A.~V. Uvarov and J.~D. Biamonte.
\newblock On barren plateaus and cost function locality in variational quantum algorithms.
\newblock \emph{Journal of Physics A: Mathematical and Theoretical}, 54\penalty0 (24):\penalty0 245301, May 2021.
\newblock ISSN 1751-8121.
\newblock \doi{10.1088/1751-8121/abfac7}.

\bibitem[Williams(1992)]{williams_simple_1992}
Ronald~J. Williams.
\newblock Simple statistical gradient-following algorithms for connectionist reinforcement learning.
\newblock \emph{Machine Learning}, 8\penalty0 (3):\penalty0 229--256, May 1992.
\newblock ISSN 1573-0565.
\newblock \doi{10.1007/BF00992696}.
\newblock URL \url{https://doi.org/10.1007/BF00992696}.

\end{thebibliography}

\end{document}